\documentclass[journal]{IEEEtran}
\usepackage[normalem]{ulem}
\usepackage[noadjust]{cite}
\usepackage{bbm}
\usepackage{array}
\usepackage{dcolumn}
\usepackage{epsfig}
\usepackage[intlimits]{amsmath}
\usepackage{amsmath, amsfonts, yhmath, bm}
\usepackage{amssymb}
\usepackage{psfrag}
\usepackage{color,soul}
\usepackage{xcolor,cancel}
\usepackage[normalem]{ulem}
\usepackage{enumerate}
\usepackage{stackengine}
\usepackage[noadjust]{cite}
\usepackage{graphicx}
\usepackage{caption}
\usepackage{subcaption}
\usepackage{multirow}
\usepackage{cite}
\usepackage[font=footnotesize]{caption}
\usepackage{etoolbox}
\usepackage{tcolorbox}
\usepackage{float}
\usepackage{amsthm}
\usepackage[ruled,vlined]{algorithm2e}
\newcommand*{\myfontb}{\fontfamily{lmr}\selectfont}

\makeatletter
\newcommand*{\coloneq}{\mathrel{\rlap{%
			\raisebox{0.3ex}{$\m@th\cdot$}}%
		\raisebox{-0.3ex}{$\m@th\cdot$}}%
	=}
\makeatother

\newcommand {\myvec}[1] {{\mbox{\boldmath $#1$}}}
\newcommand {\mymat}[1]  {{\mbox{\boldmath $#1$}}}
\newcommand {\myten}[1]  {{\mathbfcal{#1}}}

\DeclareMathAlphabet\mathbfcal{OMS}{cmsy}{b}{n}

\newenvironment{packed_enum}{
	\begin{enumerate}[(i)]
		\setlength{\itemsep}{0.5pt}
		\setlength{\parskip}{0.5pt}
		\setlength{\parsep}{0.5pt}
	}{\end{enumerate}}

\newcommand{\etal}{\textit{et al.}}

\newcommand {\A} {\mymat{A}}

\newcommand {\W} {\mymat{W}}

\newcommand {\U} {\mymat{U}}

\newcommand {\mLambda} {\mymat{\Lambda}}
\newcommand {\hA} {\widehat{\A}}
\newcommand {\bA} {\mybar{\A}}

\newcommand {\B} {\mymat{B}}
\newcommand {\J} {\mymat{J}}

\newcommand {\tA} {\widetilde{\A}}

\newcommand {\C} {\mymat{C}}
\newcommand {\D} {\mymat{D}}
\newcommand {\Z} {\mymat{Z}}
\newcommand {\mGamma} {\mymat{\Gamma}}
\newcommand {\tmGamma} {\widetilde{\mGamma}}

\newcommand {\E} {\mymat{E}}

\newcommand {\Ep} {\mymat{\mathcal{E}}}

\newcommand {\tC} {\widetilde{\C}}
\newcommand {\M} {\mymat{M}}
\newcommand {\uvarep} {\myvec{\varepsilon}}

\renewcommand {\P} {\mymat{P}}
\newcommand {\tP} {\widetilde{\P}}

\newcommand {\Q} {\mymat{Q}}

\newcommand {\R} {\mymat{R}}
\newcommand {\tR} {\widetilde{\R}}
\newcommand {\hR} {\widehat{\R}}

\newcommand {\I} {\mymat{I}}
\newcommand {\X} {\mymat{X}}

\newcommand {\F} {\mymat{F}}

\newcommand {\ue} {\myvec{e}}

\newcommand {\ua} {\myvec{a}}
\newcommand {\uba} {\mybar{\ua}}

\newcommand {\um} {\myvec{m}}

\newcommand {\uc} {\myvec{c}}
\newcommand {\ux} {\myvec{x}}

\newcommand {\ur} {\myvec{r}}

\newcommand {\up} {\myvec{p}}

\newcommand {\uv} {\myvec{v}}

\newcommand {\uo} {\myvec{0}}
\newcommand {\us} {\myvec{s}}

\newcommand {\uy} {\myvec{y}}

\newcommand {\uvarphi} {\myvec{\varphi}}

\newcommand {\utheta} {\myvec{\theta}}
\newcommand {\uvartheta} {\myvec{\vartheta}}
\newcommand {\tuvartheta} {\widetilde{\uvartheta}}
\newcommand {\hutheta} {\widehat{\myvec{\theta}}}

\newcommand {\ep} {\mathcal{E}}

\newcommand {\Rset} {\mathbb{R}}
\newcommand {\Cset} {\mathbb{C}}

\newcommand {\Eset} {\mathbb{E}}
\newcommand {\Nset} {\mathbb{N}}

\newcommand {\Diag} {\text{\normalfont Diag}}
\newcommand {\Tr} {\text{\normalfont Tr}}
\newcommand {\tps} {\rm{T}}

\newcommand {\LS} {\text{\tiny LS}}
\newcommand {\ML} {\text{\tiny ML}}
\newcommand {\OWNLLS} {\text{\tiny OWNLLS}}
\newcommand {\EJD} {\text{\tiny EJD}}
\newcommand {\KL} {\text{\tiny KL}}
\newcommand {\KLD} {\text{\tiny KLD}}

\newcommand {\tenX} {\myten{X}}
\newcommand {\tenR} {\myten{R}}
\newcommand {\htenR} {\widehat{\tenR}}

\newcommand {\cpd} {\text{\tiny CPD}}
\newcommand {\argmin} {\rm{argmin}}
\newcommand {\argmax} {\rm{argmax}}
\newcommand\norm[1]{\left\lVert#1\right\rVert}
\newcommand {\ttheta} {\widetilde{\theta}}
\newcommand {\tutheta} {\widetilde{\utheta}}
\newcommand {\htheta} {\widehat{\theta}}
\newcommand {\tua} {\widetilde{\ua}}
\newcommand {\hua} {\widehat{\ua}}
\newcommand {\huvarphi} {\widehat{\uvarphi}}
\newcommand {\hur} {\widehat{\ur}}

\DeclareMathOperator{\atantwo}{atan2}

\makeatletter
\newsavebox\myboxA
\newsavebox\myboxB
\newlength\mylenA

\newcommand*\mybar[2][0.75]{%
	\sbox{\myboxA}{$\m@th#2$}%
	\setbox\myboxB\null
	\ht\myboxB=\ht\myboxA%
	\dp\myboxB=\dp\myboxA%
	\wd\myboxB=#1\wd\myboxA
	\sbox\myboxB{$\m@th\overline{\copy\myboxB}$}
	\setlength\mylenA{\the\wd\myboxA}
	\addtolength\mylenA{-\the\wd\myboxB}%
	\ifdim\wd\myboxB<\wd\myboxA%
	\rlap{\hskip 0.5\mylenA\usebox\myboxB}{\usebox\myboxA}%
	\else
	\hskip -0.5\mylenA\rlap{\usebox\myboxA}{\hskip 0.5\mylenA\usebox\myboxB}%
	\fi}
\makeatother

\newtheorem{thm}{Theorem}

\theoremstyle{definition}
\newtheorem{definition}{Definition}

\begin{document}

\title{Blind Direction-of-Arrival Estimation in Acoustic Vector-Sensor Arrays via Tensor Decomposition and Kullback-Leibler Divergence Covariance Fitting}

\author{Amir Weiss
	
	\thanks{The author is with the Dept. of Computer Science and Applied Mathematics, Faculty of Mathematics and Computer Science, Weizmann Institute of Science,
		234 Herzl Street, Rehovot 7610001, Israel, e-mail:
		amir.weiss@weizmann.ac.il.}
}

\maketitle

\begin{abstract}
A blind Direction-of-Arrivals (DOAs) estimate of narrowband signals for Acoustic Vector-Sensor (AVS) arrays is proposed. Building upon the special structure of the signal measured by an AVS, we show that the covariance matrix of all the received signals from the array admits a natural low-rank 4-way tensor representation. Thus, rather than estimating the DOAs directly from the raw data, our estimate arises from the unique parametric Canonical Polyadic Decomposition (CPD) of the observations' Second-Order Statistics (SOSs) tensor. By exploiting results from fundamental statistics and the recently re-emerging tensor theory, we derive a consistent blind CPD-based DOAs estimate without prior assumptions on the array configuration. We show that this estimate is a solution to an equivalent approximate joint diagonalization problem, and propose an \textit{ad-hoc} iterative solution. Additionally, we derive the Cram\'er-Rao lower bound for Gaussian signals, and use it to derive the iterative Fisher scoring algorithm for the computation of the Maximum Likelihood Estimate (MLE) in this particular signal model. We then show that the MLE for the Gaussian model can in fact be used to obtain improved DOAs estimates for non-Gaussian signals as well (under mild conditions), which are optimal under the Kullback-Leibler divergence covariance fitting criterion, harnessing additional information encapsulated in the SOSs. Our analytical results are corroborated by simulation experiments in various scenarios, which also demonstrate the considerable improved accuracy w.r.t.\ Zhang \etal's state-of-the-art blind estimate \cite{zhang2012trilinear} for AVS arrays, reducing the resulting root mean squared error by up to more than an order of magnitude.
\end{abstract}
\vspace{-0.3cm}
\vspace{-0.3cm}
\section{Introduction}\label{sec:intro}
An Acoustic Vector-Sensor (AVS) is comprised of an omni-directional microphone and two/three particle velocity transducers aligned along the orthogonal coordinate axes in a two/three dimensional space \cite{nehorai1994acoustic}. It measures the acoustic pressure as well as the acoustic particle velocities in each one of the coordinate axes, thus providing a full description of the acoustic field at a given location in space \cite{de1996mu,yntema2006four}. As one would expect, in the context of passive array-processing, AVS arrays enable enhanced capabilities relative to their equivalent-aperture traditional (scalar) microphone arrays \cite{hawkes1999effects,cao2016acoustic,krishnaprasad2017doa}. Consequently, and since these devices are already practically feasible \cite{de2007microflown}, AVS arrays are the cornerstone in a wide variety of emerging applications such as aircraft acoustic detection, localization and tracking \cite{de2011detection,lo2017flight}, battlefield acoustics classification \cite{de2011acoustic}, and underwater acoustic communication \cite{song2011experimental} and source enumeration \cite{wu2014source}, to name but a few.

Within many of these applications, Direction-of-Arrival (DOA) estimation is a fundamental task, which plays a key role in the overall successful operation (e.g., localization, which reduces to DOA and range estimation). Of course, as one of the most prominent signal processing problems in general, and specifically in array processing, DOA estimation in the context of AVS arrays has already been extensively addressed in the literature during the past two and a half decades, as briefly reviewed in what follows.
\vspace{-0.5cm}
\subsection{Previous Work: DOA Estimation with AVSs}\label{subsec:previouswork}
In their seminal paper \cite{nehorai1994acoustic}, Nehorai and Paldi derived the Cram\'er-Rao Lower Bound (CRLB) on the Mean Squared Error (MSE) of any unbiased DOAs estimate, and proposed two algorithms---the Intensity-Based and the Velocity-Covariance-Based Algorithms---for DOA estimation, though for a single source with a single AVS only. For this (limited) scenario, the maximum steered response power and the Maximum Likelihood Estimate (MLE, as a special case of the former), were derived by Levin \etal\ in \cite{levin2011direction} and \cite{levin2012maximum}, resp. Multisource DOA estimation in a reverberant environment, still using a single AVS, has been recently addressed by Wu \etal\ \cite{wu2018multisource}, where low-reverberant-single-source
points in the time-frequency domain are exploited.

For the extended Multiple-Sources Multiple-Sensors (MSMS) scenario, Hawkes and Nehorai considered in \cite{hawkes1998acoustic} both the conventional and the minimum-variance distortionless response beamforming DOAs estimates to demonstrate the improvement attained by using AVSs rather than traditional pressure sensors. Following this work, an abundance of methods have been proposed for various specific scenarios (whether for array- or signal-related properties), such as linear \cite{chen2004wideband}, circular \cite{zou2009circular,gur2017modal,shi2019real}, sparse \cite{li2010maximum,yuan2012coherent,rao2015doa} and nested \cite{pal2010nested,han2013improved,han2014direction} arrays, coherent signals \cite{chen2005coherent,yuan2010modified,palanisamy2012two}, one-bit measurements \cite{ramamohan2019blind} and a variety of others \cite{li2012improved,zhong2012particle,ramamohan2018uniaxial}. However, all these methods require perfect or (at least) partial prior knowledge of the array configuration, or, equivalently\footnote{Under some conventional, reasonable assumptions and/or approximations, such as a signal planar wavefront (``far-field") approximation.}, of the steering vectors parametric structure, in particular w.r.t.\ the sources' DOAs. Hence, these methods are typically sensitive to model inaccuracies/errors of this sort. This is exactly the point where our work comes into play w.r.t.\ to our novel contributions.

In this work, we consider blind DOAs estimation in AVS arrays within the framework of narrowband signals. Here, the term ``blind" implies that no prior assumption on the array configuration (/geometry) is used, hence the steering vectors parametric structure is deemed unknown. In fact, blind DOA estimation in the context of AVSs has so far been only sparsely addressed in the literature. A direction-finding and blind interference rejection algorithm was proposed by Wong in \cite{wong2010acoustic}, but only for up to three fast frequency-hop spread spectrum signals of unknown hop sequences, and using only a single AVS. Xiao \etal\ also proposed a blind DOA estimate for a single AVS only \cite{xiao2013blind}, based on the Joint Approximate Diagonalization of Eigen-matrices (JADE, \cite{cardoso1993blind}) algorithm, and therefore does not provide a solution for Gaussian signals. 

For AVS arrays, Zhang \etal\ proposed the Trilinear Decomposition-based (TriD) blind DOAs estimate for incoherent signals \cite{zhang2012trilinear}, and a similar TriD approach in \cite{zhang2013paralind} for coherent signals. However, the trilinear decomposition considered both in \cite{zhang2012trilinear} and \cite{zhang2013paralind} is exact only for the (less practical) noiseless signal model. Although the TriD methods perform quite well, the DOA estimates which stem from this approach were not shown analytically to be consistent or optimal. In contrast, on top of providing superior performance, our novel estimate is shown analytically (and demonstrated empirically) to be consistent \textit{regardless} of the SNR conditions, and provides optimal performance for Gaussian signals for \textit{any} SNR.
\vspace{-0.25cm}
\subsection{Blind DOAs Estimation: Motivation and Contributions}\label{subsec:motivationblind}
The motivation to address this blind scenario and develop a solution algorithm under this framework arises from several considerations. Firstly, by not assuming a specific array configuration (geometry/structure), the resulting solution algorithm could properly operate in different systems with arbitrary array configurations, and without the need for specific tuning prior to operation. Secondly, unknown inter-AVSs' gain and/or phase offsets (e.g., due to sensors mis-locations), which otherwise require a calibration procedure (e.g., \cite{song2014quasi}), are totally transparent to such a blind estimate. For example, many of the DOA estimation methods tailored to Uniform Linear Arrays (ULAs) exploit the special Toeplitz structure of the observations' spatial covariance matrix (e.g., \cite{paulraj1985direction} as one representative example). Clearly, the performance of such methods deteriorates rapidly in the presence of sensors error positioning, in contrast to the performance of a blind estimate, which remains (almost) indifferent to these errors. In other words, a blind algorithm is robust w.r.t.\ sensors error positioning. This robustness implies a significant practical advantage over other algorithms in terms of simplicity for the end-user, and may save time and resources, for example, when blind (``online") calibration is not possible and \textit{ad-hoc} transmission of a calibrating source is required. Thirdly, most algorithms are developed for an underlying signal model based on some physical approximation, e.g., the ``near-field" or ``far-field" approximations (e.g., \cite{tichavsky2001near} in the context of AVSs). As a result, in these cases the steering vectors' parametric representation (in particular, w.r.t.\ the DOAs) is also only an approximation, which brings along with analytical convenience an inherent modeling error \cite{tam2009cramer,tam2014hybrid}. This modeling error impairs the performance even for an optimal solution (in some well-defined sense) for this particular, approximated model. In contrast, a blind approach, in which the solution is not developed based on such (potential) modeling errors, would yield model-based errors free\footnote{In the respect explained above regarding the steering vectors' parametric structure, not entirely for all possible model-based errors.} DOA estimates. Lastly, a blind estimate can also successfully cope with faulty elements in the array \cite{liu2018robustness}, and maintains proper functionality for partially damaged arrays.

Motivated by the merits above, in this work we propose a novel blind DOA estimation algorithm for AVS arrays with arbitrary configurations, while making only a few \textit{a-priori} assumptions on the signal model. Building upon the AVS measurement model, we exploit the special (block) structure of the observations' covariance matrix, which naturally lends itself to the recently flourishing tensor formulation in the signal processing literature (e.g., \cite{comon2014tensors,han2014nested,cichocki2015tensor,balda2016first,sidiropoulos2017tensor,yeredor2018high,kanatsoulis2018hyperspectral,yeredor2019maximum}). As a natural continuum thereof, we employ a statistical approach and invoke tensor-calculus-related results, leading to a consistent DOAs estimate for \textit{any} SNR. We then show that this estimate can be improved by a second refinement phase via Kullback-Leibler Divergence (KLD) covariance fitting, which yields our proposed estimate. The main contributions of this paper are as follows:
\begin{itemize}
	\item \emph{Consistent blind DOAs estimation via tensor decomposition:} Based on the observations' empirical covariance matrix and a consistent noise variance estimate, we show that joint estimation of all the DOAs and their associated steering vectors in a MSMS scenario is (asymptotically) equivalent to a parametric Canonical Polyadic Decomposition (CPD, see \cite{domanov2013uniqueness} and reference therein)---sometimes termed as ``tensor rank decomposition" or ``parallel factor model"---of a $4$-mode tensor statistic. We show that the uniqueness theorem of quadrilinear decompositions of $4$-mode arrays due to Sidiropoulos and Bro (Theorem $2$ in \cite{sidiropoulos2000uniqueness}) grants this CPD-based estimate its consistency.
	\item \emph{Iterative solution algorithm of the CPD-based estimate:} We show that computation the CPD-based estimate amounts to a partially-parametric Approximate Joint Diagonalization (AJD) problem. Accordingly, we proposed an iterative solution, which is a modified version of the Alternating Columns-Diagonal Centers (AC-DC) algorithm proposed by Yeredor in \cite{yeredor2002non}. Our modified algorithm results in more accurate DOAs estimates than the ones a generic AJD algorithm would yield due to our tailored parametric adaptation. Further, the algorithm inherently yields estimates of the steering vectors, considered as nuisance parameters, which consequently also enable consistent blind separation of the latent sources.
	\item \emph{Performance bounds and optimal estimation for Gaussian signals:} For the particular case of Gaussian signals, we derive the CRLB on the MSE matrix of any unbiased estimate in joint estimation of \emph{all} the unknown deterministic model parameters, namely the steering vectors, the DOAs and the noise variance. In addition, based on the Fisher Information Matrix (FIM), and using the (already obtained) CPD-based consistent estimates as initial solutions, we propose a Maximum Likelihood (ML) ``refinement" phase, in which all the unknown parameters' MLEs are pursued via the Fisher Scoring Algorithm (FSA, \cite{jennrich1976newton}). As demonstrated in simulations, these refined estimates are asymptotically efficient, attaining the CRLB.
	\item \emph{KLD covariance fitting enhancement:} We show that the MLEs for the particular Gaussian signal model are also optimal under the KLD covariance fitting criterion \emph{regardless} of the underlying signal model, and can therefore be used for non-Gaussian signals as well in order to achieve significant performance enhancement. Accordingly, our final proposed DOAs estimate enjoys higher accuracy and robustness to the underlying signal model.
\end{itemize}

The rest of this paper is organized as follows. The following subsection contains an outline of our notations. In Section \ref{sec:problemformulation} we present the model under consideration and formulate our blind DOA estimation problem. The CPD-based (phase 1) estimates are presented in Section \ref{sec:proposedblindestimate}, followed by the iterative solution algorithm for their actual computation in Section \ref{sec:ModACDCalg}. We then consider in Section \ref{subsec:CRBound} the Gaussian signal model, and derive its respective CRLB, as well as the update equations of the FSA for the computation of the MLEs. Our proposed KLD-based (phase 2) estimates are presented in Section \ref{sec:KLDcovFit}, followed by simulation results in Section \ref{sec:simulationresults}, substantiating and demonstrating empirically our analytical results. Concluding remarks are given in Section \ref{sec:conclusion}.
\vspace{-0.3cm}
{\subsection{Notations and Preliminaries}\label{subsec:notations}
We use $x, \ux$, $\X$ and $\tenX$ for a scalar, column vector, matrix and tensor, resp. The superscripts $(\cdot)^{\tps}$, $(\cdot)^*$, $(\cdot)^{\dagger}$, $(\cdot)^{-1}$ and $(\cdot)^{+}$ denote the transposition, complex conjugation, conjugate transposition, inverse and Moore-Penrose pseudo-inverse operators, resp. We use $\I_{K}$ to denote the $K\times K$ identity matrix, and the pinning vector $\ue_k$ denotes the $k$-th column of the identity matrix with context-dependent dimension. Further, $\delta_{k\ell}\triangleq\ue_k^{\tps}\ue_{\ell}$ denotes the Kronecker delta of $k$ and $\ell$. $\Eset[\cdot]$ denotes expectation, the $\Diag(\cdot)$ operator forms an $M\times M$ diagonal matrix from its $M$-dimensional vector argument, and $\uo_M\in\Rset^{M\times 1}$ is the all-zeros vector. The Kronecker, Khatri-Rao (column-wise Kronecker) and tensor outer products (e.g., {\cite{sidiropoulos2017tensor}}) are denoted by $\otimes, \diamond$ and $\circ$, resp. We use $\jmath$ (a dotless $j$) to denote $\sqrt{-1}$; The operators $\Re\{\cdot\}$ and $\Im\{\cdot\}$ denote the real and imaginary parts (resp.) of their complex-valued argument. $\text{rank}(\Q)$ denotes the rank of the matrix $\Q$. The Frobenius and $\ell^2$ norms are denoted by $\norm{\cdot}_{\rm{F}}$ and $\norm{\cdot}_{2}$, resp. Convergence in probability and in distribution are denoted by $\xrightarrow[\quad\;]{p}, \xrightarrow[\quad\;]{d}$, resp., as $T\rightarrow\infty$, where $T$ denotes the sample size. 

The function $\atantwo(y,x)$ returns the principal value of the argument function applied to the complex number $x+\jmath y$. The gradient of a matrix function $\F(x)\in\Cset^{M\times N}$  w.r.t.\ its scalar argument $x\in\Rset$ is denoted by $\nabla_x\F\in\Cset^{M\times N}$. Conversely, the gradient of a scalar function $x(\F)\in\Rset$ w.r.t.\ its matrix argument $\F\in\Cset^{M\times N}$ is denoted by $\nabla_{\text{\boldmath$F$}}x\in\Cset^{M\times N}$. The $\text{vec}(\cdot)$ operator concatenates the columns of an $M \times N$ matrix into an $MN \times 1$ column vector. The $\text{vec}^*\left(\cdot\right)$ operator, defined only for Hermitian matrices, is the invertible transformation which concatenates the columns of its $M\times M$ Hermitian matrix argument into an $(M(M+1)/2)\times 1$ column vector, but takes each element (conjugately) ``duplicated" by conjugate symmetry only once, on its first occurrence. $\text{cum}(w,x,y,z)$ denotes the fourth-order joint cumulant of its four scalar random variable arguments. Finally, as we make use of the somewhat less familiar Kruskal Rank, for convenience, we bring its definition, given as follows.
\begin{definition}
{[\textit{Kruskal Rank}]} \textit{Let $\Q\in\Cset^{I\times D}$. The Kruskal Rank of $\Q$, denoted by $k_{\text{\boldmath $Q$}}$, is $r$ if and only if every $r$ columns are linearly independent, and this fails for at least one set of $r+1$ columns. It follows that $k_{\text{\boldmath $Q$}}\leq{\emph{rank}}(\Q)\leq\emph{min}(I,D)$.}
\end{definition}	
\vspace{-0.4cm}
\section{Problem Formulation}\label{sec:problemformulation}
Consider an array of $M$ AVSs, where each AVS consists of three elements, one pressure and two particle velocity transducers in two perpendicular directions. The configuration of the array, which is not constrained to a particular structure (e.g., uniform linear), is assumed as unknown, which dictates a ``blind" setup in this respect. Further, consider the presence of $D<M-1$ unknown narrowband sources, centered around some common carrier frequency with a wavelength $\lambda$, where we assume that the number of sources $D$ is known. Assuming the received signals are down-converted, Low-Pass Filtered (LPF)\footnote{The bandwidth of the LPF exceeds the bandwidth of the widest source.} and sampled at least at the Nyquist rate, the vector of sampled baseband signals from all $3M$ sensors is given by
\begin{equation}\label{modelequation}
\uy[t]=\bA(\utheta)\us[t]+\uv[t]\triangleq\ux[t]+\uv[t]\in\Cset^{3M\times1},
\end{equation}
for all $t\in\{1,\ldots,T\}$, where
\begin{packed_enum}
	\item $\us[t]\hspace{-0.05cm}\triangleq\hspace{-0.05cm}\left[s_1[t]\,\cdots\,s_D[t]\right]^{\tps}\hspace{-0.05cm}\in\hspace{-0.05cm}\Cset^{D\times1}$ is the vector of sources impinging on the array from unknown azimuth angles $\utheta\triangleq\left[\theta_1\;\cdots\;\theta_D\right]^{\tps}\in[-\pi,\pi)^{D\times1}$, assumed as distinct from one another, i.e, $\forall d\neq\ell:\theta_d\neq\theta_\ell$;
	\item $\bA(\utheta)\triangleq\left[\uba(\theta_1)\,\cdots\,\uba(\theta_D)\right]\in\Cset^{3M\times D}$ is the array manifold matrix, whose columns are the steering vectors
	\begin{equation}\label{steervecdef}
	\uba(\theta_d)\hspace{-0.025cm}\triangleq\hspace{-0.025cm}[\underbrace{\ua^{\tps}(\theta_d)}_{\substack{\text{pressure} \\ \text{sensors}}}\;\underbrace{\cos(\theta_d)\ua^{\tps}(\theta_d)}_{\substack{x\text{-velocity} \\ \text{sensors}}}\;\underbrace{\sin(\theta_d)\ua^{\tps}(\theta_d)}_{\substack{y\text{-velocity} \\ \text{sensors}}}]^{\tps},
	\end{equation}
	in which $\{\ua(\theta_d)\in\Cset^{M\times1}\}_{d=1}^D$
	are the unknown equivalent acoustic pressure sensor array steering vectors \cite{nehorai1994acoustic};
	\item $\uv[t]\in\Cset^{3M\times1}$ is an additive noise vector, spatially and temporally independent, identically distributed (i.i.d.)\ zero-mean circular Complex Normal (CN) \cite{loesch2013cramer} with a covariance matrix $\R_v\triangleq\Eset\left[\uv[t]\uv[t]^{\dagger}\right]=\sigma^2_v\I_{3M}$, where $\sigma^2_v\in\Rset_+$ is assumed as (deterministic) unknown\footnote{Following \cite{hawkes2001acoustic}, we absorb the factor modeling the noise difference between the pressure and velocity channels in the mixing matrix parameters.}; and
	\item $\ux[t]$ is the signal that would have been received in the absence of the additive noise $\uv[t]$, namely with $\sigma_v^2=0$.
\end{packed_enum}

We also assume that the sources may be modeled as temporally i.i.d.\ proper (\hspace{1sp}\cite{neeser1993proper}) zero-mean mutually uncorrelated stochastic processes, statistically independent of the noise $\uv[t]$. We denote the sources' unknown diagonal covariance matrix as $\R_s\triangleq\Eset\left[\us[t]\us[t]^{\dagger}\right]\in\Rset_+^{D\times D}$. Hence,
\begin{equation}\label{covaiancematrixofr}
\begin{aligned}
\hspace{-0.225cm}\R_y&\triangleq\Eset\left[\uy[t]\uy[t]^{\dagger}\right]=\Eset\left[\ux[t]\ux[t]^{\dagger}\right]+\Eset\left[\uv[t]\uv[t]^{\dagger}\right]\\
&\triangleq\hspace{-0.015cm}\R_x\hspace{-0.015cm}+\hspace{-0.015cm}\R_v\hspace{-0.015cm}=\hspace{-0.015cm}\bA(\utheta)\R_s\bA(\utheta)^{\dagger}\hspace{-0.015cm}+\hspace{-0.015cm}\sigma_v^2\I_{3M}\hspace{-0.015cm}\in\hspace{-0.015cm}\Cset^{3M\times3M},
\end{aligned}
\end{equation}
where $\R_x$ is the covariance matrix of the noiseless signal $\ux[t]$.

Thus, the problem at hand can be formulated as follows:
\tcbset{colframe=gray!95!blue,size=small,width=0.49\textwidth,arc=2.1mm,outer arc=1mm}
\begin{tcolorbox}[upperbox=visible,colback=white]
	\textbf{Problem:} {\myfontb\emph{Given the i.i.d.\ measurements $\left\{\uy[t]\right\}_{t=1}^{T}$, without prior knowledge of the parametric structure of $\ua(\theta_d)$, estimate the unknown DOAs $\{\theta_1,\ldots,\theta_D\}$.}}
\end{tcolorbox}
\vspace{-0.5cm}
\section{Phase 1: The CPD-based Blind DOAs Estimate}\label{sec:proposedblindestimate}
Let us begin with a general, bird's-eye view description of our strategy for the proposed solution, which stems from two fundamental observations. First, observe that since the array configuration is assumed unknown, the ``core" (pressure) steering vectors $\{\ua(\theta_d)\}_{d=1}^D$ are unknown, and therefore $\{\uba(\theta_d)\}_{d=1}^D$ are as well. However, the {\myfontb\emph{structure}} of each $\uba(\theta_d)$ as a function of $\ua(\theta_d)$ given in \eqref{steervecdef}, which is determined inherently by the AVS's basic structure, is not only known, but also encapsulates the dependence on the desired DOA $\theta_d$ (via $\cos(\theta_d)$ and $\sin(\theta_d)$) {\myfontb\emph{regardless}} of the particular parametric structure of $\ua(\theta_d)$. Next, notice that while the measured signal \eqref{modelequation}, which is a function of the DOAs $\utheta$, is random for all $t\in\{1,\ldots,T\}$, the covariance matrix $\R_y$, given in \eqref{covaiancematrixofr}, is a deterministic function of $\utheta$. Therefore, we shall work towards writing the explicit dependence of $\R_y$ in $\utheta$ based on the special structure of an AVS steering vector, as prescribed by \eqref{steervecdef}. Since the empirical covariance matrix $\hR_y\triangleq\frac{1}{T}\sum_{t=1}^{T}{\uy[t]\uy[t]^{\dagger}}\in\Cset^{3M\times 3M}$ is a consistent estimate of $\R_y$ (under mild conditions), we may exploit the aforementioned special dependence to derive our blind, consistent DOAs estimate based on $\hR_y$ only.

More specifically, notice first that due to \eqref{steervecdef}, we may write the array manifold matrix $\bA(\utheta)$ as
\begin{equation}\label{AVSarraymatrixdef}
\bA(\utheta)=\begin{bmatrix}
\A(\utheta)\\
\A(\utheta)\Diag\left(\cos(\utheta)\right)\\
\A(\utheta)\Diag\left(\sin(\utheta)\right)\\
\end{bmatrix}\triangleq\C(\utheta)\diamond\A(\utheta),
\end{equation}
where $\A(\utheta)\hspace{-0.05cm}\triangleq\hspace{-0.05cm}\left[\ua(\theta_1)\,\cdots\,\ua(\theta_D)\right]\hspace{-0.05cm}\in\hspace{-0.05cm}\Cset^{M\times D}$, $\cos(\utheta)$ and $\sin(\utheta)$ operate elementwise, and we have defined the auxiliary matrix
\begin{equation}\label{auxiliarymatrixC}
\C(\utheta)\hspace{-0.075cm}\triangleq\hspace{-0.075cm}\begin{bmatrix}
1\hspace{-0.08cm}&\hspace{-0.08cm}\dots\hspace{-0.08cm}&\hspace{-0.08cm}1\\
\cos(\theta_1)\hspace{-0.08cm}&\hspace{-0.08cm}\dots\hspace{-0.08cm}&\hspace{-0.08cm}\cos(\theta_D)\\
\sin(\theta_1)\hspace{-0.08cm}&\hspace{-0.08cm}\dots\hspace{-0.08cm}&\hspace{-0.08cm}\sin(\theta_D)\end{bmatrix}\hspace{-0.08cm}\triangleq\hspace{-0.05cm}\left[\uc(\theta_1)\,\cdots\,\uc(\theta_D)\right]\hspace{-0.07cm}\in\hspace{-0.07cm}\Rset^{3\times D},
\end{equation}
which is completely determined by $\utheta$ only. Since in our framework the ``core" steering vectors' parametric structure $\{\ua(\theta_d)\}_{d=1}^D$ is assumed unknown, we denote for brevity hereafter $\ua(\theta_d):=\ua_d$ for all $d\in\{1,\ldots D\}$, and $\A:=\A(\utheta)$ accordingly. Note, however, that we intentionally keep the notation $\bA(\utheta)$, as its dependence on $\utheta$ via $\C(\utheta)$, regardless of $\{\ua_d\}_{d=1}^D$, is known and given by \eqref{AVSarraymatrixdef}. Note further that although $\bA(\utheta)$ has $3MD$ complex-valued elements, \eqref{AVSarraymatrixdef} implies that it is completely determined only by the unknowns $\{\ua_d,\theta_d\}_{d=1}^D$, namely $2MD+D$ free parameters (/ degrees of freedom). This special, economical structure will be exploited shortly.

Next, let us consider the covariance matrix $\R_x$ of the noiseless signal $\ux[t]$. For this, observe first that both $\A$ and the sources powers, i.e., the diagonal elements of the (diagonal) matrix $\R_s$, are unknown. Thus, without loss of generality\footnote{{\myfontb\emph{Scaling}} of the sources is an inherent ambiguity in such a blind scenario.} (w.l.o.g.), we may assume that $\R_s=\I_{D}$. With this, $\R_x$ reads
\begin{align}\label{covarianceofx}
\hspace{-0.1cm}\R_x&=\bA(\utheta)\R_s\bA(\utheta)^{\dagger}=\big(\C(\utheta)\diamond\A\big)\big(\C(\utheta)\diamond\A\big)^{\dagger}\\
&=\sum_{d=1}^{D}{\left(\uc(\theta_d)\otimes\ua_d\right)\left(\uc(\theta_d)\otimes\ua_d\right)^{\dagger}}\label{mixedproductrule}\\
&=\sum_{d=1}^{D}{\left(\uc(\theta_d)\uc(\theta_d)^{\tps}\right)\otimes\left(\ua_d\ua_d^{\dagger}\right)}\triangleq\sum_{d=1}^{D}{\F(\theta_d)\otimes\A_d},\label{blockmatrixRx}
\end{align}
where we have used the mixed product rule (e.g., \cite{sidiropoulos2017tensor}, Section II) in moving from \eqref{mixedproductrule} to \eqref{blockmatrixRx}, and defined the rank-1 matrices
\begin{equation}\label{Fthetamatrix}
\F(\theta_d)=\begin{bmatrix}
1&\cos(\theta_d)&\sin(\theta_d)\\
\cos(\theta_d)&\cos^2(\theta_d)&\tfrac{1}{2}\sin(2\theta_d)\\
\sin(\theta_d)&\tfrac{1}{2}\sin(2\theta_d)&\sin^2(\theta_d)\end{bmatrix}\in\Rset^{3\times3}
\end{equation}
and $\A_d=\ua_d\ua_d^{\dagger}\in\Cset^{M\times M}$ for all $d\in\{1,\ldots,D\}$. Now, observe that writing \eqref{blockmatrixRx} explicitly, we have
\begin{equation}\label{blcokmatrixcovexplicitly}
\R_x=\sum_{d=1}^{D}{\begin{bmatrix}
F_{11}(\theta_d)\A_d&F_{12}(\theta_d)\A_d&F_{13}(\theta_d)\A_d\\
F_{21}(\theta_d)\A_d&F_{22}(\theta_d)\A_d&F_{23}(\theta_d)\A_d\\
F_{31}(\theta_d)\A_d&F_{32}(\theta_d)\A_d&F_{33}(\theta_d)\A_d\end{bmatrix}},
\end{equation}
which leads to the natural definition of the $4$-mode covariance tensor $\tenR_x\in\Cset^{3\times 3\times M\times M}$, with an $(i,j)$-th dorsal slab \cite{sidiropoulos2017tensor}
\begin{equation}
\label{dorsalslabdef}
\tenR_x(i,j,:,:)\hspace{-0.05cm}\triangleq\hspace{-0.05cm}\sum_{d=1}^{D}{\hspace{-0.025cm}F_{ij}(\theta_d)\A_d}\hspace{-0.05cm}=\hspace{-0.05cm}\sum_{d=1}^{D}{\hspace{-0.025cm}c_i(\theta_d)c_j(\theta_d)\A_d}\hspace{-0.05cm}\in\hspace{-0.05cm}\Cset^{M\times M},
\end{equation}
for all $i,j\in\{1,2,3\}$. Indeed, since $\{\A_d\}$ are rank-1 matrices, one can also write the $(i,j,m,n)$-th element of $\tenR_x$ as,
\begin{equation}\label{covtensorrankD2}
\begin{gathered}
\mathcal{R}_x(i,j,m,n)=\sum_{d=1}^{D}{c_i(\theta_d)c_j(\theta_d)A_{md}A^*_{nd}}\\
\Longrightarrow\;\tenR_x=\sum_{d=1}^{D}{\uc(\theta_d)\circ\uc(\theta_d)\circ\ua_d\circ\ua_d^*},
\end{gathered}
\end{equation}
to conclude that $\tenR_x$ is a $4$-mode tensor of rank $D$.

\textit{Remarks:}
\begin{enumerate}[i.]
	\item Notice that while $M$, the number of AVSs in the array, may be very large, the number of sources $D$ is typically smaller. Hence, since $\tenR_x$ is a rank-$D$ tenor, it admits a CPD of order $D$, and is therefore another compact, economical (and equivalent) representation of $\R_x$.
	\item The dorsal slab $\tenR_x(i,j,:,:)$ is the auto- / cross-covariance matrix between similar / different types of sensors. Thus, $\tenR_x(i,i,:,:)$ are the auto-covariance matrices of all types of sensors (e.g., $\tenR_x(2,2,:,:)$ is the auto-covariance matrix of the $x$-velocity sensors). Similarly, $\tenR_x(i,j,:,:)$, for $i\neq j$, are the cross-covariance matrices between different types of sensors (e.g., $\tenR_x(1,3,:,:)$ is the cross-covariance matrix between the pressure and $y$-velocity sensors).
	\item The tensor $\tenR_x$ depends on (and is determined by) $\utheta$ even without knowledge of the explicit parametric dependence of $\{\ua_d\}_{d=1}^D$ in $\utheta$, i.e., when treating $\A$ as a general unknown complex-valued matrix with $MD$ elements, which are $2MD$ free parameters (/ degrees of freedom).
\end{enumerate}

A well-known appealing property of tensors is the {\myfontb\emph{uniqueness}} of their CPD. This powerful property holds under relatively mild conditions, as formulated in the following theorem for $4$-mode tensors due to Sidiropoulos and Bro \cite{sidiropoulos2000uniqueness}:
\begin{thm}\label{tensoruniqueCPD}
\emph{[\textit{CPD Uniqueness of $4$-mode Tensors}]} Consider the $D$-component $4$-mode tensor
\begin{equation}\label{fourmodetensorthm}
\begin{gathered}
\mathcal{X}(i,j,m,n)=\sum_{d=1}^{D}{Q_{id}U_{jd}W_{md}Z_{nd}}\in\Cset,\\
\forall i\in\{1,\ldots,I\},\;\forall j\in\{1,\ldots,J\},\\
\forall m\in\{1,\ldots,M\}, \forall n\in\{1,\ldots,N\},
\end{gathered}
\end{equation}
with $\Q\in\Cset^{I\times D}, \U\in\Cset^{J\times D}, \W\in\Cset^{M\times D}$ and $\Z\in\Cset^{N\times D}$, and suppose that $\mathcal{X}(i,j,m,n)$ cannot be represented using fewer than $D$ components as in \eqref{fourmodetensorthm}. Then, given $\tenX$, the factor matrices $\Q,\U,\W$ and $\Z$ are unique up to permutation and complex scaling of their columns provided that
\begin{equation}\label{kruskalcond}
k_{\text{\boldmath $Q$}}+k_{\text{\boldmath $U$}}+k_{\text{\boldmath $W$}}+k_{\text{\boldmath $Z$}}\geq 2D+3.
\end{equation}
\end{thm}
\noindent The proof of Theorem \ref{tensoruniqueCPD} is given in \cite{sidiropoulos2000uniqueness}, Section 3.

Going back to our problem, in order to avoid the immaterial permutation and scaling ambiguities, which have no effect whatsoever on the DOAs estimation, we assume w.l.o.g.\ that the first element of each steering vector is non-negative, i.e., $A_{1d}\in\Rset_{\geq0}$ for all $d\in\{1,\ldots,D\}$, and that the DOAs are ordered in an ascending order, i.e., $\theta_1<\ldots<\theta_D$. Furthermore, we now assume that $\text{rank}(\A)=D$, and refer to this condition as the {\myfontb\emph{sensor array regularity condition}}. Recall that $\theta_d\neq\theta_\ell$ for all $d\neq\ell\in\{1,\ldots,D\}$ by assumption, hence the sensor array regularity condition is quite mild, and typically holds for any reasonable array configuration / geometry. This regularity condition grants the tensor $\tenR_x$ in \eqref{covtensorrankD2} its uniqueness, as we show in the following theorem.
\begin{thm}\label{covtensoruniqeCPD}
\emph{[\textit{CPD Uniqueness of $\tenR_x$}]} Consider the $4$-mode tensor $\tenR_x$ as defined in \eqref{dorsalslabdef}. Assume $D>1$, and that the sensor array regularity condition holds, i.e., $\emph{rank}(\A)=D$. Then, the CPD \eqref{covtensorrankD2} of the rank-$D$ tensor $\tenR_x$ is unique.
\end{thm}
\begin{proof}
See Appendix \ref{AppA}.
\end{proof}
At this point, based on the uniqueness of the CPD \eqref{covtensorrankD2} of the covariance tensor $\tenR_x$, we are ready to present the blind CPD-based DOAs estimate, whose intuitive definition comes naturally from Theorem \ref{covtensoruniqeCPD}. Given the measurements $\{\uy[t]\}_{t=1}^T$,
\begin{enumerate}
\item Compute the empirical covariance matrix $\hR_y$;
\item Estimate the noise variance $\sigma_v^2$ via the MLE for Gaussian signals \cite{wax1985detection} or the improved (less biased) estimate \cite{kritchman2009non}. Denote it as $\widehat{\sigma}_v^2$ (note that this is a consistent estimate);
\item Define the covariance matrix estimate of $\ux[t]$, and the corresponding block partitioning (according to \eqref{blcokmatrixcovexplicitly})
\begin{equation}\label{covmatofxest}
\hR_x\triangleq\hR_y-\widehat{\sigma}_v^2\I_{3M}\triangleq{\begin{bmatrix}
\hR_x^{(1,1)}&\hR_x^{(1,2)}&\hR_x^{(1,3)}\\
\hR_x^{(2,1)}&\hR_x^{(2,2)}&\hR_x^{(2,3)}\\
\hR_x^{(3,1)}&\hR_x^{(3,2)}&\hR_x^{(3,3)}\end{bmatrix}},
\end{equation}
and using \eqref{covmatofxest}, construct the estimated $4$-mode covariance tensor $\htenR_x\in\Cset^{3\times 3\times M\times M}$ of $\ux[t]$,
\begin{equation}\label{covtenofxest}
\htenR_x(i,j,:,:)\triangleq\hR_x^{(i,j)}\in\Cset^{M\times M},\; \forall i,j\in\{1,2,3\};
\end{equation}
\item Given the statistic tensor $\htenR_x$, jointly estimate the DOAs vector $\utheta$ and the ``core" steering vectors matrix $\A$ via
\tcbset{colframe=gray!95!blue,size=small,width=0.45\textwidth,arc=2.1mm,outer arc=1mm}
\begin{tcolorbox}[upperbox=visible,colback=white]
\begin{equation}\label{CPDDOAestimate}
\left(\hutheta_{\cpd},\hA_{\cpd}\right)\triangleq\underset{\substack{\text{{\boldmath $\theta$}$\in[-\pi,\pi)^{D\times1}$} \\ \text{{\boldmath $\A$}$\in\Cset^{M\times D}$}}}{\argmin} \norm{\tenR\left(\utheta,\A\right)-\htenR_x}^2_{\rm{F}},
\end{equation}
\end{tcolorbox}
where $\tenR\left(\tutheta,\tA\right)$ is the parametric tensor function
\begin{equation}\label{parametrictensorfunction}
\begin{gathered}
\tenR:[-\pi,\pi)^{D\times1}\times\Cset^{M\times D}\rightarrow\Cset^{M\times M\times D\times D},\\
\tenR\left(\tutheta,\tA\right)\triangleq\sum_{d=1}^{D}{\uc(\ttheta_d)\circ\uc(\ttheta_d)\circ\tua_d\circ\tua_d^*},
\end{gathered}
\end{equation}
with $\tutheta\in[-\pi,\pi)^{D\times1}$ and $\tA\triangleq\left[\tua_1\;\cdots\;\tua_D\right]\in\Cset^{M\times D}$, and we use $\widetilde{\;\,}$ to emphasize that, in general, the arguments may be different from the {\myfontb\emph{true}} unknown estimands $\utheta, \A$. 
\end{enumerate}

Obviously, due to Theorem \ref{covtensoruniqeCPD}, given the {\myfontb\emph{true}} covariance tensor $\tenR_x$, which admits the {\myfontb\emph{exact}} CPD \eqref{covtensorrankD2}, the sources' DOAs are readily extracted from the vectors $\{\uc(\theta_d)\}_{d=1}^D$ via
\begin{equation}\label{DOAexaccomputation}
\atantwo\left(c_3(\theta_d),c_2(\theta_d)\right)=\theta_d, \; \forall d\in\{1,\ldots,D\}.
\end{equation}
Therefore, intuitively, if $\hR_y$ and $\widehat{\sigma}_v^2$ are ``good" estimates, than so does $\htenR_x$, and it makes sense to define the estimate \eqref{CPDDOAestimate}: the best approximate CPD of $\htenR_x$ in the Least Squares (LS) sense yields approximate versions of $\{\uc(\theta_d)\}_{d=1}^D$, from which the DOAs' estimates arise. Fortunately, this rationale may be rigorously justified, as we prove in the following theorem.
\begin{thm}\label{consistencyofCPDestimate}
\emph{[\textit{Consistency of the CPD-based estimates $\hutheta_{\emph{\cpd}}$ and $\hA_{\emph{\cpd}}$}]} Let $\hR_y$ and $\widehat{\sigma}_v^2$ be consistent estimates of $\R_y$ and $\sigma_v^2$, resp. Further, assume that all the elements of $\A$ are finite, such that 
\begin{equation}\label{boundingspherecondition}
\exists\rho\in\Rset_+:\norm{\A}_{\rm{F}}\leq\rho,
\end{equation}
and that the DOAs vector $\utheta$ belong to a compact set, such that
\begin{equation}\label{DOAscompactsetcondition}
\exists\epsilon\in\Rset_+:\utheta\in[-\pi,\pi-\epsilon]^{D\times 1}.
\end{equation}
Then,\hspace{-0.05cm} the estimates $\hutheta_{\emph{\cpd}}$ and $\hA_{\emph{\cpd}}$ defined in \eqref{CPDDOAestimate} are consistent, i.e., with any fixed SNR level, for a sample size $T\rightarrow\infty$,
\begin{equation}\label{consistencyofCPDestpconvergence}
\left(\hutheta_{\emph{\cpd}},\hA_{\emph{\cpd}}\right)\xrightarrow[\quad\;]{p}\left(\utheta,\A\right).
\end{equation}
\end{thm}
\begin{proof}\label{proofCPDestconsistency}
See Appendix \ref{AppB}.
\end{proof}
Notice that the conditions Theorem \ref{consistencyofCPDestimate} requires are quite mild, and are fulfilled in practice for the most part. For distributions with finite fourth order moments, $\hR_y$ and $\widehat{\sigma}_v^2$ are consistent estimates (where the noise variance is estimated as in \cite{wax1985detection} or \cite{kritchman2009non}). Furthermore, in practice, the steering vectors $\{\ua_d\}_{d=1}^D$ are always finite. Lastly, condition \eqref{DOAscompactsetcondition} is required for technical considerations in the proof presented in Appendix \ref{AppB}, but is meaningless from a practical point of view for a sufficiently small $\epsilon$. Therefore, we conclude that under these mild conditions, and without knowledge of the explicit dependence of $\A$ in $\utheta$, \eqref{CPDDOAestimate} are consistent. In particular, we derived $\hutheta_{\cpd}$, consistent blind DOAs estimates, as desired.

We note in passing that our approach also yield, as a by product, the nuisance parameters' estimate $\hA_{\cpd}$, allowing for consistent separation of the latent sources $\us[t]$, by multiplying $\hA^+_{\cpd}$, the pseudo-inverse of $\hA_{\cpd}$, to the left of $\uy[t]$.

Having provided the blind DOAs estimates \eqref{CPDDOAestimate} and the analytical guarantees for their consistency, we now turn to present an iterative algorithm for their actual computation.
\vspace{-0.2cm}
\section{Computation of the CPD-based Estimates via the Modified AC-DC Algorithm}\label{sec:ModACDCalg}
Our goal in this section is to find a computationally feasible algorithm in order to obtain the proposed estimates $\hutheta_{\cpd}$ and $\hA_{\cpd}$, given the estimated covariance tensor $\htenR_x$. The roadmap towards this goal is the following. First, we show that the optimization problem \eqref{CPDDOAestimate} is in fact equivalent to an AJD problem, with underlying diagonal matrices which admit a particular parametric structure. Then, we resort to the AJD AC-DC algorithm \cite{yeredor2002non}, which due to its alternating mode of operation, allows for a local convenient modification, tailored \textit{ad-hoc} to the particular parametric structure of the aforementioned diagonal matrices. By this, we obtain an iterative solution algorithm for the optimization problem \eqref{CPDDOAestimate}.

For the first step, let us define the LS cost function
\begin{equation}\label{costfunctionLS}
C_{\LS}\left(\utheta,\A\right)\triangleq\norm{\tenR\left(\utheta,\A\right)-\htenR_x}^2_{\rm{F}}\in\Rset_+.
\end{equation}
Now, since we may write 
\begin{equation}\label{diagonalizationformulation}
\sum_{d=1}^{D}{F_{ij}(\theta_d)\A_d}=\A\,\Diag\left(F_{ij}(\utheta)\right)\A^{\dagger}\triangleq\A\D_{ij}(\utheta)\A^{\dagger},
\end{equation}
where $F_{ij}(\utheta)$ is elementwise (e.g., $F_{23}(\utheta)=\tfrac{1}{2}\sin(2\utheta)$) and $\D_{ij}(\utheta)\in\Rset^{D\times D}$, using \eqref{blcokmatrixcovexplicitly}--\eqref{dorsalslabdef} and \eqref{covtenofxest}, observe that 
\begin{align}\label{costfunctionLSAJD}
C_{\LS}\left(\utheta,\A\right)&=\sum_{i,j=1}^{3}{\norm{\sum_{d=1}^{D}{F_{ij}(\theta_d)\A_d}-\htenR_x(:,:,i,j)}^2_{\rm{F}}}\\
&=\sum_{i,j=1}^{3}{\norm{\A\D_{ij}(\utheta)\A^{\dagger}-\hR_x^{(i,j)}}^2_{\rm{F}}}.\label{costfunctionLSAJD2}
\end{align}
Therefore, the optimization problem in \eqref{CPDDOAestimate} is equivalent to AJD in the LS sense. More specifically, since $\A$ is not restricted to a particular structure (e.g., orthogonal matrix), \eqref{costfunctionLSAJD2} accounts for a {\myfontb\emph{non-orthogonal}} AJD problem.

\subsection{Review of the AC-DC Algorithm}\label{subsec:reviewACDC}
One viable solution approach for such a non-orthogonal AJD problem, is using the iterative AC-DC algorithm, proposed by Yeredor \cite{yeredor2002non}. In a nutshell, given a set of $K$ ``target-matrices" $\{\Q_k\in\Cset^{N\times N}\}_{k=1}^K$, the algorithm seeks a ``diagonalizing matrix" $\B\in\Cset^{N\times L}$ and $K$ associated diagonal matrices $\{\mLambda_k\in\Cset^{L\times L}\}_{k=1}^K$, such that
\begin{equation}\label{ACDCcostfunction}
C_{\text{\tiny AC-DC}}\left(\{\mLambda_k\}_{k=1}^K,\B\right)\triangleq\sum_{k=1}^{K}{w_k\norm{\B\mLambda_k\B^{\dagger}-\Q_k}_{\rm{F}}^2}
\end{equation}
is minimized, where $\{w_k\in\Rset_+\}_{k=1}^K$ are some positive weights. The algorithm alternates between the two following minimization schemes:
\begin{itemize}
\item  The AC (``alternating columns") phase minimizes $C_{\text{\tiny AC-DC}}$ w.r.t.\ a single column of $\B$ while keeping its other columns, as well as $\{\mLambda_k\}_{k=1}^K$, fixed. This phase is sequentially repeated, for all columns of $\B$, for a prespecified number of ``sweeps".
\item The DC (``diagonal centers") phase minimizes $C_{\text{\tiny AC-DC}}$ w.r.t.\ the diagonal matrices $\{\mLambda_k\}_{k=1}^K$ while keeping $\B$ fixed.
\end{itemize}

It is readily seen that \eqref{costfunctionLSAJD2} admits the same formulation as \eqref{ACDCcostfunction}, with the simple mapping\footnote{The order of the six distinct pairs $(i,j)$ mapped to the index $k$ is insignificant, as long as all six distinct options are mapped.} 
\begin{equation}\label{AJDisomorphism}
\begin{tabular}{ c c c }
\underline{AC-DC} &  & \underline{DOAs via CPD}\vspace{0.15cm} \\
$(N,L,K)$ & $\rightarrow$ & $(M,D,6)$ \\
$k\in\{1,\ldots,6\}$ & $\rightarrow$ & $(i,j)\in\mathcal{I}_3\triangleq\{(\ell,p):\ell\leq p\leq3\}_{\ell=1}^3$ \\
$w_k$ & $\rightarrow$ & $2-\delta_{ij}$ \\
$\B$ & $\rightarrow$ & $\A$ \\
$\left\{\mLambda_k\right\}$ & $\rightarrow$ & $\left\{\D_{ij}(\utheta)\right\}$ \\
$\left\{\Q_k\right\}$ & $\rightarrow$ & $\{\hR_x^{(i,j)}\}$
\end{tabular}
\end{equation}
where we have used $\hR_x^{(i,j)}=\hR_x^{(j,i)}$ for all $i,j\in\{1,2,3\}$. Yet, although it is possible to apply the AC-DC algorithm in its original form in order to solve our specific CPD optimization problem, a significant enhancement can be achieved by exploiting a subtle difference in these two (almost) identical problem, as follows. While the AC-DC attempts to minimize \eqref{ACDCcostfunction} for some general diagonal matrices $\{\mLambda_k\}$, namely with $KL$ free parameters (corresponding to $6D$ parameters according to \eqref{AJDisomorphism}) in $\Cset$, in our problem, the diagonal matrices $\{\D_{ij}(\utheta)\}$ are parametrized by the DOAs vector $\utheta$, namely by only $D$ free parameters in $[-\pi,\pi)$. Therefore, it would be desirable if certain modifications in the AC-DC algorithm could be made, such that the optimization w.r.t.\ the diagonal matrices would actually be only w.r.t.\ $\utheta$. The reason for this is twofold. Firstly, the optimization would be for less parameters---$D$ rather than $6D$. Secondly, the optimization for each $\theta_d$ would be confined to the interval $[-\pi,\pi)$, rather than searching on the whole complex plane. Consequently, not only the computational cost would be reduced, but the variance in the resulting diagonal elements' estimates would be reduced as well, due to their known parametric structure \eqref{diagonalizationformulation}.

Fortunately, the AC-DC algorithm operates iteratively in an alternating manner between two phases: While in the AC phase the diagonalizing matrix $\B$ is optimized with the diagonal matrices $\{\mLambda_k\}$ kept fixed, the optimization in the DC phase is w.r.t.\ $\{\mLambda_k\}$ {\myfontb\emph{only}} with $\B$ kept fixed. Therefore, we shall now redesign the DC phase so as to adjust it to the specifics of our problem, in order to enjoy the aforementioned advantages.
\vspace{-0.2cm}
\subsection{The Modified DC Phase}\label{subsec:modDCphase}
As mentioned above, in terms of our problem, during this phase the current estimate of $\A$ is held fixed. Therefore, in order to emphasize that this is not the final estimate $\hA_{\cpd}$, we denote this (intermediate) estimate as $\hA$ for brevity. Hence, the problem under consideration in this phase is as follows
\begin{equation}\label{costAJDwrttheta}
\hutheta=\underset{\text{{\boldmath $\theta$}$\in[-\pi,\pi)^{D\times1}$}}{\argmin}C_{\LS}\left(\utheta,\hA\right),
\end{equation}
where $\hutheta$ denotes the (intermediate) estimate of $\utheta$ in the modified DC phase. Since optimizing \eqref{costAJDwrttheta} is a non-convex $D$-dimensional optimization problem, we take a similar approach as in the original AC phase, and minimize $C_{\LS}$ w.r.t.\ $\theta_d$ while keeping all other DOAs $\{\theta_\ell\}_{\ell\neq d}$ fixed. The complete modified DC phase is then comprised of a predefined, fixed number of sweeps over all the DOAs $\{\theta_d\}_{d=1}^D$. Thus, the relaxed optimization problem at hand is now
\begin{align}\label{costAJDwrtsingletheta}
\htheta_d&=\underset{\theta_d\in[-\pi,\pi)}{\argmin}C_{\LS}\left(\left[\htheta_1\,\cdots\,\theta_d\,\cdots\,\htheta_D\right]^{\tps},\hA\right)\\
&\triangleq\underset{\theta_d\in[-\pi,\pi)}{\argmin}\widetilde{C}_{\LS}(\theta_d),\label{costAJDwrtsingletheta2}
\end{align}
namely a trigonometrical scalar function of a real-valued scalar argument. As we show in the reminder of this subsection, \eqref{costAJDwrtsingletheta2} may be solved efficiently, thus leading to our desired goal in deriving a modified DC phase, tailored specifically to our primary {\myfontb\emph{parametric}} CPD optimization problem.

Starting our derivation, in Appendix \ref{AppC} we show that differentiating $\widetilde{C}_{\LS}$ w.r.t.\ $\theta_d$ yields after algebraic simplifications
\begin{equation}\label{DOALSequation}
\frac{\partial \widetilde{C}_{\LS}}{\partial\theta_d}=\alpha\cos(\theta_d)-\beta\sin(\theta_d)+\gamma\cos(2\theta_d)-\delta\sin(2\theta_d),
\end{equation}
where
\begin{align}
\alpha&=4\left(\sum_{\substack{k=1 \\ k\neq d}}^{D}{\left|\hua_d^{\dagger}\hua_k\right|\sin(\htheta_k)}-\hua_d^{\dagger}\hR_x^{(1,3)}\hua_d\hspace{-0.075cm}\right),\label{constantLSequation1}\\
\beta&=4\left(\sum_{\substack{k=1 \\ k\neq d}}^{D}{\left|\hua_d^{\dagger}\hua_k\right|\cos(\htheta_k)}-\hua_d^{\dagger}\hR_x^{(1,2)}\hua_d\hspace{-0.075cm}\right),\label{constantLSequation2}\\
\gamma&=2\left(\sum_{\substack{k=1 \\ k\neq d}}^{D}{\left|\hua_d^{\dagger}\hua_k\right|\sin(2\htheta_k)}-2\hua_d^{\dagger}\hR_x^{(2,3)}\hua_d\hspace{-0.075cm}\right),\label{constantLSequation3}\\
\delta&=2\left(\sum_{\substack{k=1 \\ k\neq d}}^{D}{\left|\hua_d^{\dagger}\hua_k\right|\cos(2\htheta_k)}-\hua_d^{\dagger}\left(\hR_x^{(2,2)}-\hR_x^{(3,3)}\right)\hua_d\hspace{-0.075cm}\right)\label{constantLSequation4}
\end{align}
are constants w.r.t.\ the estimand $\theta_d$. Note that \eqref{constantLSequation1}--\eqref{constantLSequation4} depend on the index $d$, which is not reflected in their notations for the sake of brevity. Granted, the global minimizer of $\widetilde{C}_{\LS}$ is, in particular, a stationary point of $\widetilde{C}_{\LS}$. Hence, it is necessarily a solution of the equation $\tfrac{\partial \widetilde{C}_{\LS}}{\partial\theta_d}=0$, namely
\begin{equation}\label{thetaLSequations}
\alpha\cos(\theta_d)-\beta\sin(\theta_d)+\gamma\cos(2\theta_d)-\delta\sin(2\theta_d)=0.
\end{equation}
Introducing the transformation $\tau\triangleq\tan\left(\frac{\theta_d}{2}\right)$
enables us to rewrite \eqref{thetaLSequations} in terms of the variable $\tau$ as (see Appendix \ref{AppC})
\begin{equation}\label{transformedDOALSequation}
(3\gamma+\alpha)\tau^4+2\beta\tau^3+2\gamma\tau^2+(4\delta+2\beta)\tau-(\alpha+\gamma)=0.
\end{equation}
Evidently, \eqref{transformedDOALSequation} is a $4$-th order polynomial equation in $\tau$, which has at most four real-valued solutions, and may be solved efficiently using various methods (e.g., \cite{strobach2010fast}). Thus, the global minimizer \eqref{costAJDwrtsingletheta2}, denoted as $\theta_*$, is the solution that minimizes $\widetilde{C}_{\LS}$ out of these (maximum four) solutions, and is the updated estimate of $\theta_d$. We stress that this approach does not guarantee that the output of the modified AC-DC algorithm would be the global maximizer \eqref{CPDDOAestimate}. However, it does guarantee that $C_{\LS}$ is non-increasing w.r.t.\ the iterations, exactly as for the original AC-DC algorithm, thus retaining its original weak convergence property (see \cite{yeredor2002non}, Section IV). Summarizing the above, the modified DC phase is given in Algorithm \ref{Algorithm1}.
\begin{algorithm}[t]
\KwIn{$\hA,\hutheta,N_{\text{s}}$ (current estimates\hspace{0.06cm}+\hspace{0.06cm}number of sweeps)}
\KwOut{$\hutheta_{\text{\tiny MOD-DC}}$ (updated estimate of $\utheta$)}
\nl \ForAll{$N_{\text{\emph{s}}}$ \emph{sweeps}}{\nl \For{$d=1,\ldots,D$}{\nl Compute $\alpha,\beta,\gamma,\delta$ via \eqref{constantLSequation1}--\eqref{constantLSequation4}, resp.;\\
\nl Solve \eqref{transformedDOALSequation}, and translate only the real-valued solutions according to $\theta_d=2\tan^{-1}(\tau)$;\\
\nl Evaluate $\widetilde{C}_{\LS}(\theta_d)$ for each of the solutions. Denote the global minimizer as $\theta_*$;\\
\nl Update $\htheta_d=\theta_*$;}}
\nl return $\hutheta_{\text{\tiny MOD-DC}}=\hutheta$.
\caption{{\bf Modified DC Phase} \label{Algorithm1}}
\end{algorithm}
\setlength{\textfloatsep}{2pt}

Note that due to the required matrix multiplications, evaluating $\widetilde{C}_{\LS}(\theta_d)$ at four points amounts to $\mathcal{O}\left(M^3\right)$ operations. Since the overall computational load per iteration\footnote{Considered as a single run of the AC/DC phase, with a single full sweep.} of the AC-DC algorithm, in terms of the left hand side of \eqref{AJDisomorphism}, is $\mathcal{O}\left(KN^3\right)$ (see \cite{yeredor2002non}, Section IV), it is $\mathcal{O}\left(6M^3\right)$ in terms of the parameters of our problem, according to the right hand side of \eqref{AJDisomorphism}. Therefore, we conclude that the overall computational load per iteration of the AC-DC algorithm with the modified DC phase remain $\mathcal{O}\left(M^3\right)$, namely unchanged.
\vspace{-0.3cm}
\subsection{Initialization via Exact Joint Diagonalization}\label{subsec:initModACDC}
Naturally, whenever an iterative algorithm is proposed the issue of initialization must be addressed. To this end, recall first that \eqref{CPDDOAestimate} is a solution to the equivalent AJD problem \eqref{costfunctionLSAJD2}. Therefore, we propose to use Yeredor's Exact Joint Diagonalization (EJD) for AJD method \cite{yeredor2005using}, given in our case by the following steps:
\begin{enumerate}
	\item Construct vectors $\um_x^{(i,j)}\triangleq\text{svec}\left(\hR_x^{(i,j)}\right), \forall(i,j)\in\mathcal{I}_3$;
	\item Construct the matrix $\M_x\triangleq\sum_{(i,j)\in\mathcal{I}_3}{\um_x^{(i,j)}\left(\um_x^{(i,j)}\right)^{\tps}}$;
	\item Find the two largest eigenvalues and associated eigenvectors $\up_1$ and $\up_2$ of $\M_x$;
	\item Construct the matrices $\P_k\triangleq\text{unsvec}\left(\up_k\right)$, $k=1,2$; and
	\item Compute the eigenvalue decomposition of $\P_1\P_2^{-1}$, and denote $\hA_{\EJD}\in\Cset^{M\times D}$ as the matrix with the $D$ eigenvectors corresponding to the $D$ largest eigenvalues.
\end{enumerate}
For the definitions of the $\text{svec}(\cdot)$ and $\text{unsvec}(\cdot)$ operators, see \cite{yeredor2005using}, Subsection II-B. Having obtained the EJD-based solution $\hA_{\EJD}$, we further carry out the following additional steps:
\begin{enumerate}
	\setcounter{enumi}{5}
	\item $\hA_{\EJD}\leftarrow \hA_{\EJD}\cdot\Diag\left(\left[e^{-\jmath\phi_{11}^{\EJD}}\cdots e^{-\jmath\phi_{1D}^{\EJD}}\right]\right)$, where $e^{\jmath\phi_{1d}^{\EJD}}\triangleq \widehat{A}_{\EJD_{1d}}/|\widehat{A}_{\EJD_{1d}}|$ for all $d\in\{1,\ldots,D\}$, such that the first element of each steering vector is real-valued, as required;
	\item Compute the pair of estimated diagonal matrices
	\begin{equation*}
	j\in\{2,3\}:\widehat{\D}_{1j}\left(\utheta\right)\triangleq\Re\{\hA_{\EJD}^+\hR_x^{(1,j)}\hA_{\EJD}\}\in\Rset^{D\times D};
	\end{equation*}
	\item Compute $\hutheta_{\EJD}$, the initial EJD-based DOA estimates via \eqref{DOAexaccomputation}, where the estimates of $c_2(\theta_d)$, $c_3(\theta_d)$ are the associated diagonal elements of $\widehat{\D}_{12}\left(\utheta\right)$, $\widehat{\D}_{13}\left(\utheta\right)$, resp.
\end{enumerate}

Note that although this non-iterative solution is ``merely" an initial solution to the proposed modified AC-DC algorithm, it is instrumental for a successful operation in practice. Indeed, while the estimate \eqref{CPDDOAestimate} is consistent, computing it boils down to finding the global minimizer of $C_{\text{LS}}$, which, in general, is not a trivial task. However, and as we show via simulations in Section \ref{sec:simulationresults}, when the proposed initialization is used, the iterative algorithm yields the desired estimates.

As an intermediate summary, by establishing Algorithm \ref{Algorithm1}, we have thus presented a consistent blind DOAs estimate \eqref{CPDDOAestimate}, along with an iterative solution algorithm---the AC-DC algorithm with the modified DC phase---equipped with a non-iterative ``educated" initial solution, tailored \textit{ad-hod} to our specific blind AVS DOA estimation problem.

We now turn to the particular case of Gaussian signals, in which significant performance enhancement can be attained by further refining the CPD-based estimates $\hutheta_{{\cpd}}, \hA_{{\cpd}}$ from \eqref{CPDDOAestimate}. Moreover, and quite interestingly, we show that the results obtained within the Gaussian signal model framework are valid for other signal models as well, enabling this performance enhancement for a wider class of signals' distributions.
\vspace{-0.1cm}
\section{Optimal Estimation for Gaussian Signals}\label{subsec:CRBound}
In this section only, we further assume that $\us[t]$ is circular CN. As a consequence, it follows that
\begin{equation}\label{CN_samples}
\uy[t]\sim \mathcal{CN}\left(\uo_{3M},\R_y\right), \forall t\in\{1,\ldots,T\}.
\end{equation}
Since the model is now fully specified, it is first instructive to study the CRLB on the MSE of any unbiased estimate in joint estimation of all the unknown deterministic parameters, namely $\A,\utheta$ and $\sigma_v^2$ (recall that $\R_s=\I_D$ by assumption).
\vspace{-0.2cm}
\subsection{Cram\'er-Rao Lower Bound in Blind DOA Estimation}\label{subsec:CRLB}
For brevity, we define the vector of all real-valued unknowns
\begin{equation}\label{defofvarphi}
\uvarphi\triangleq\left[\text{vec}(\Re\{\A\})^{\tps}\;\text{vec}(\Im\{\widetilde{\I}_M\A\})^{\tps}\;\utheta^{\tps}\;\sigma_v^2\right]^{\tps}\in\Rset^{K_{\varphi}\times1},
\end{equation} 
where $\widetilde{\I}_M\triangleq\left[\uo_{M-1}\;\ue_{1}\;\cdots\;\ue_{M}\right]\in\Rset^{(M-1)\times M}$, when multiplies from the left, preserves all the rows except for the first one of a matrix with $M$ rows (recall $A_{1d}\in\Rset_{\geq0}$), and $K_{\varphi}\triangleq2MD+1$ is the total number of unknown parameters.

Since $\left\{\uy[t]\right\}_{t=1}^{T}$ are all CN and i.i.d., the FIM elements corresponding to $\A,\utheta$ and $\sigma_v^2$ are given by\footnote{We specifically use a different notation for the FIM's elements, with slight abuse in notation also in \eqref{FIM_example}, which is more natural w.r.t.\ their definition, and therefore easier to comprehend in this context, for the sake of clarity.} (see, e.g., \cite{collier2005fisher})
\begin{equation}\label{FIM_CN_element}
\begin{gathered}
J[\varphi_i,\varphi_j]=T\cdot{\Tr\left(\R_y^{-1}\left(\nabla_{\varphi_i}\R_y\right)\R_y^{-1}\left(\nabla_{\varphi_j}\R_y\right)\right)},\\
\forall i,j\in\{1,\ldots,K_{\varphi}\},
\end{gathered}
\end{equation}
where $\J(\uvarphi)$ denotes the FIM. In Appendix \ref{AppD} we show that
\begin{align}
\nabla_{\Re\{A_{md}\}}\R_y&=\F(\utheta)\otimes\left(\ue_m\ua_d^{\dagger}+\ua_d\ue_m^{\tps}\right),\label{gradrealA}\\
\nabla_{\Im\{A_{\tilde{m}d}\}}\R_y&=\jmath\cdot\F(\utheta)\otimes\left(\ue_{\tilde{m}}\ua_d^{\dagger}-\ua_d\ue_{\tilde{m}}^{\tps}\right),\label{gradimagA}\\
\nabla_{\theta_d}\R_y&=\nabla_{\theta_d}\F(\utheta)\otimes\A_d,\label{gradthetad}\\
\nabla_{\sigma_v^2}\R_y&=\I_{3M}\label{gradsigmav}.
\end{align}
In addition, by the Woodbury matrix identity \cite{woodbury1950inverting}, we have 
\begin{equation}
\label{FIM_term_R_inv}
\R_y^{-1}\hspace{-0.025cm}=\hspace{-0.025cm}\frac{1}{\sigma_v^2}\left[\I_{3M}\hspace{-0.025cm}-\hspace{-0.025cm}\frac{1}{\sigma_v^2}\bA(\utheta)\left( \I_D + \bA(\utheta)^{\dagger}\bA(\utheta) \right)^{-1}\bA(\utheta)^{\dagger}\right].
\end{equation}
Therefore, all the required expressions for the computation of the FIM are at hand. For example, by \eqref{gradthetad} and \eqref{gradsigmav}, we have
\begin{equation}\label{FIM_example}
J[\theta_d,\sigma_v^2] = T\cdot\Tr\left(\R_y^{-1}\left(\nabla_{\theta_d}\F(\utheta)\otimes\A_d\right)\R_y^{-1}\right).
\end{equation}
The CRLB on the MSE in unbiased joint estimation of $\A,\utheta$ and $\sigma_v^2$ is readily given by the inverse of the FIM $\J(\uvarphi)$, whose elements are prescribed in \eqref{FIM_CN_element}, using \eqref{gradrealA}--\eqref{FIM_term_R_inv}. 

Next, we further utilize this result in order to derive an approximate iterative solution algorithm for the computation of the MLE
\begin{equation}\label{GaussianMLE}
\huvarphi_{\ML}\triangleq\underset{\text{{\boldmath $\varphi$}$\in\Rset^{K_{\varphi}\times1}$}}{\argmax} \prod_{t=1}^{T}{\frac{1}{\pi^{3M}\det\left(\R_y\right)}e^{-\text{{\boldmath$y$}$[t]^{\dagger}${\boldmath $R$}$_y^{-1}${\boldmath $y$}$[t]$}}}.
\end{equation}
\vspace{-0.5cm}
\subsection{MLE Computation via the Fisher Scoring Algorithm}\label{subsec:FSA}
Given an initial estimate of $\uvarphi$, the FSA can be used in order to obtain a stationary point of the log-likelihood (if it converges). Moreover, if this initial estimate is ``close" enough to the global maximizer of the log-likelihood, the algorithm is likely to converge to the MLE. Formally, the update equation of the FSA for the $n$-th iteration is given by
\begin{equation}\label{FSAupdatequation}
\huvarphi^{(n)} = \huvarphi^{(n-1)} + \J^{-1}\left(\huvarphi^{(n-1)}\right)\left.\nabla_{\bm{\varphi}}\mathcal{L}\right|_{\scriptsize{\uvarphi}=\huvarphi^{(n-1)}},
\end{equation}
where $\huvarphi^{(n)}$ is the estimate of $\uvarphi$ in the $n$-th iteration, and
\begin{equation}\label{loglikelihood}
\mathcal{L}(\uvarphi)\triangleq -T\cdot\left(\log\det\R_y + \Tr\left(\hR_y\R_y^{-1}\right)\right)+c
\end{equation}
is the log-likelihood function, where $c$ is a constant independent of $\uvarphi$. Hence, in order to carry out the iterations \eqref{FSAupdatequation} such that they will successfully converge to the MLE $\huvarphi_{\ML}$, three ingredients are required: a sufficiently ``good" initial solution $\huvarphi^{(0)}$, and closed-form expressions of the FIM and the score function, i.e., $\J(\uvarphi)$ and $\nabla_{\bm{\varphi}}\mathcal{L}$, resp.

Now, recall that $\hA_{\cpd},\hutheta_{\cpd}$ and $\widehat{\sigma}_v^2$, specified in \eqref{CPDDOAestimate} and \eqref{covmatofxest}, resp., are consistent estimates. Therefore, when rearranged in vector form according to \eqref{defofvarphi},
\begin{equation}\label{initialestimateFSA}
\huvarphi^{(0)}\triangleq\left[\text{vec}(\Re\{\hA_{\cpd}\})^{\tps}\;\text{vec}(\Im\{\widetilde{\I}_M\hA_{\cpd}\})^{\tps}\;\hutheta_{\cpd}^{\tps}\;\widehat{\sigma}_v^2\right]^{\tps}
\end{equation}
can serve as a ``good" initial estimate of $\uvarphi$, which is presumably ``close" to the global maximizer of \eqref{loglikelihood}. Furthermore, note that we have already obtained closed-form expressions for the elements of the FIM $\J(\uvarphi)$, given in \eqref{FIM_CN_element} while using \eqref{gradrealA}--\eqref{FIM_term_R_inv}, which can be computed for any $\uvarphi$. Moreover, using the chain rule, we have for all $i\in\{1,\ldots,K_{\varphi}\}$
\begin{equation}\label{score_wrt_varphi}
\nabla_{\varphi_i}\mathcal{L}=\sum_{k,\ell=1}^{3M}{\frac{\partial\mathcal{L}\left(\uvarphi\right)}{\partial R_{y_{k\ell}}}\cdot\frac{\partial R_{y_{k\ell}}}{\partial \varphi_i}}=\Tr\left(\nabla_{\bm{R}_y}\mathcal{L}\cdot\nabla_{\varphi_i}\R_y^{\tps}\right),
\end{equation}
where $\nabla_{\bm{R}_y}\mathcal{L}\in\Cset^{3M\times 3M}$ and $\nabla_{\varphi_i}\R_y\in\Cset^{3M\times 3M}$. As we show in Appendix \ref{AppD},
\begin{equation}\label{score_wrt_covmatR}
\nabla_{\bm{R}_y}\mathcal{L}=-T\cdot\left[\R_y^{-1}\left(\I_{3M}-\hR_y\R_y^{-1}\right)\right]^{\tps},
\end{equation}
and with the already obtained expressions \eqref{gradrealA}--\eqref{gradsigmav}, we have obtained a closed-form expression of the score \eqref{score_wrt_varphi} w.r.t.\ any element of $\uvarphi$, such that $\nabla_{\bm{\varphi}}\mathcal{L}$ can be computed for any $\uvarphi$. 

By this, we now have the three required ingredients for a successful operation of the FSA for the computation of the MLE $\huvarphi_{\ML}$. Note further that as long as $M\cdot D$ is not too ``large" (in terms of matrix inversion), the computation of $\J^{-1}\left(\huvarphi^{(n)}\right)$ is not very costly w.r.t.\ computational load.
\vspace{-0.3cm}
\section{Phase 2: The KLD-based Blind DOAs Estimate}\label{sec:KLDcovFit}
Our goal in this section is to show that \eqref{FSAupdatequation} can successfully operate not only for Gaussian signals. Moreover, we would like to achieve analytical arguments which explain and justify two aspects. First, the rationale of invoking results obtained under the Gaussian model for non-Gaussian signals. And second, the resulting higher accuracy attained with this approach for our primary goal, DOAs estimation. To this end, we begin by showing that, in general, it is theoretically possible to attain better estimates, in terms of lower MSEs, than the CPD-based estimates \eqref{CPDDOAestimate} with only mild assumptions on the signal model.
\vspace{-0.7cm}
\subsection{Suboptimality of the Equally-Weighted LS Estimates}\label{subsec:suboptimality}
In our general framework, as long as the fourth-order joint cumulants of the measurements $\{\uy[t]\}_{t=1}^T$ are finite, i.e.,
\begin{equation}\label{finite4thordercumulant}
\begin{gathered}
\exists \varrho\in\Rset_+: \forall i,j,k,\ell\in\{1,\ldots,M\}:\\
|\kappa_y[i,j,k,\ell]|\triangleq|\text{cum}(y_i[t],y_j^*[t],y_k[t],y^*_\ell[t])|<\varrho,
\end{gathered}
\end{equation}
the estimate $\hR_y$, which is {\myfontb\emph{not}} necessarily the MLE of $\R_y$, is still consistent by virtue of the law of large numbers \cite{ross2006first}. Moreover, denoting $\ur_y\triangleq\text{vec}^*\left(\R_y\right)\in\Cset^{K_r\times1}$ with $K_r\triangleq3M(3M+1)/2$, by virtue of the central limit theorem\footnote{Note that the diagonal elements of $\hR_y$ are in fact real-valued, and are therefore normal, rather than CN. However, this does not affect the following derivation, so we allow this slight abuse in notation.} \cite{ross2006first},
\begin{equation}\label{CLT_cov_mat}
\text{vec}^*\left(\hR_y\right)\triangleq\hur_y\xrightarrow[\quad\;]{d}\mathcal{CN}\left(\ur_y,\mGamma_{\varepsilon},\C_{\varepsilon}\right),
\end{equation}
where $\mGamma_{\varepsilon}\in\Cset^{K_r\times K_r}$ and $\C_{\varepsilon}\in\Cset^{K_r\times K_r}$ are the covariance and pseudo-covariance matrices of $\hur_y$, resp. Note that $\ur_y, \mGamma_{\varepsilon}$ and $\C_{\varepsilon}$ all depend on the unknown parameters vector $\uvarphi$.

Now, since the asymptotic distribution of $\hur_y$, or, equivalently, $\hR_y$, is given by \eqref{CLT_cov_mat}, as we show in Appendix \ref{AppE}, ML estimation of $\uvarphi$ based on $\hR_y$ asymptotically amounts to
\begin{equation}\label{AsymptoticMLEviaOWLS}
\huvarphi_{\OWNLLS}\triangleq\underset{\text{{\boldmath $\varphi$}$\in\Rset^{K_{\varphi}\times1}$}}{\argmin} \left[\uvarep^{\dagger}\quad\uvarep^{\tps}\right]\R_{\varepsilon}^{-1}\left[\uvarep^{\tps}\quad\uvarep^{\dagger}\right]^{\tps}\in\Rset_+,
\end{equation}
where $\uvarep\triangleq\hur_y-\ur_y$ denotes the vector of estimation errors $\Ep\triangleq\hR_y-\R_y$ in estimation of $\R_y$, i.e., $\uvarep=\text{vec}^*\left(\Ep\right)$, and
\begin{equation}\label{augmentedCovariance4asymptoticMLE}
\R_{\varepsilon}\triangleq\Eset\left[\begin{bmatrix}
\uvarep\\
\uvarep^*\end{bmatrix}\left[\uvarep^{\dagger}\;\;\uvarep^{\tps}\right]\right]=\begin{bmatrix}
\mGamma_{\varepsilon} & \C_{\varepsilon} \\
\C_{\varepsilon}^* & \mGamma_{\varepsilon}^*\end{bmatrix}\in\Cset^{2K_r\times2K_r}.
\end{equation}
Evidently, as seen from \eqref{AsymptoticMLEviaOWLS}, asymptotically optimal estimation of the unknown parameters based only on $\hR_y$ is via the Optimally-Weighted Non-Linear LS (OWNLLS) criterion. Note that the optimal weight matrix $\R_{\varepsilon}^{-1}$ is a particular case of a general (not necessarily optimal) positive-definite weight matrix, denoted, say, by $\W$. Indeed, with equal weights, namely $\W=\I_{2K_r}$, the criterion yields the CPD-based estimates \eqref{CPDDOAestimate}. However, as we show in Appendix \ref{AppE}, the optimal weight matrix $\R_{\varepsilon}^{-1}$ is generally \emph{not} equal to a scaled identity matrix. Therefore, in general, the equally-weighted LS criterion \eqref{CPDDOAestimate} yields suboptimal estimates. Thus, we conclude that even for the general (non-Gaussian) signal model, pursuing improved, more accurate estimates than $\hutheta_{\cpd}$ and $\hA_{\cpd}$ is not in vain.

We stress that the estimate \eqref{AsymptoticMLEviaOWLS} is generally not the MLE w.r.t.\ the raw data $\{\uy[t]\}_{t=1}^T$, since $\hR_y$ is not necessarily a sufficient statistic. Nevertheless, despite the possible strict statistical insufficiency of $\hR_y$, it still encapsulates valuable information which is not fully used via equally-weighted LS fitting \eqref{CPDDOAestimate}, and can be further exploited, as we show next.
\vspace{-0.3cm}
\subsection{Enhancement via the KLD Covariance Fitting Criterion}\label{subsec:KLDcovfittingNonGaussian}
Resorting temporarily to the framework of Gaussian signals, as explained in Subsection \ref{subsec:FSA}, the FSA \eqref{FSAupdatequation} asymptotically yield the MLE\footnote{If initialized in the basin of attraction of the global maximizer of \eqref{loglikelihood}.} $\huvarphi_{\ML}$. As can be easily seen from the log-likelihood \eqref{loglikelihood}, the sufficient statistic in the Gaussian model is the sample covariance matrix $\hR_y$. Therefore, ML estimation of $\uvarphi$ based {\myfontb \emph{only}} on $\hR_y$ yields the MLE of $\uvarphi$ based on the raw data $\{\uy[t]\}_{t=1}^T$. Moreover, as shown in Appendix \ref{AppE}, ML estimation of $\uvarphi$ based on $\hR_y$ asymptotically amount to OWNLLS estimation of $\uvarphi$ based on $\hR_y$ \eqref{AsymptoticMLEviaOWLS}. Consequently, it follows that for the CN signal model \eqref{CN_samples}, asymptotically,
\begin{equation}\label{OWNLLSisMLE}
\huvarphi_{\ML}\approx\huvarphi_{\OWNLLS}.
\end{equation}
Accordingly, this implies that the iterations \eqref{FSAupdatequation}, namely the FSA for the Gaussian model, are (approximately) an implicit computation of the OWNLLS estimate \eqref{AsymptoticMLEviaOWLS}. In other words, maximizing the log-likelihood \eqref{loglikelihood} is asymptotically equivalent to minimizing the OWNLLS objective in \eqref{AsymptoticMLEviaOWLS}. Hence \eqref{FSAupdatequation} is an iterative procedure for a solution of either of these equivalent objectives. Furthermore, note that in this model, the optimal weight matrix $\R_{\varepsilon}^{-1}$ is computed under the Gaussian model \eqref{CN_samples} as well, and is thus fully specified by the covariances and pseudo-covariances (for all $i,j,k,\ell\in\{1,\ldots,M\}$)
\begin{equation}\label{covandpseudocovGaussian}
\Eset\left[\ep_{ij}\ep_{k\ell}^*\right]=\frac{1}{T}R_{y_{ik}}R^*_{y_{j\ell}},\quad \Eset\left[\ep_{ij}\ep_{k\ell}\right]=\frac{1}{T}R_{y_{i\ell}}R^*_{y_{jk}}.
\end{equation}
For the detailed derivation of \eqref{covandpseudocovGaussian}, see Appendix \ref{AppE}.

At this point, observe that the log-likelihood \eqref{loglikelihood} reads
\begin{equation}\label{loglikelihoodisKLD}
\begin{gathered}
\mathcal{L}(\uvarphi)=-T\cdot\left(\log\det\R_y + \Tr\left(\hR_y\R_y^{-1}\right)\right) + c\\
=-T\cdot\left(\log\left(\frac{\det\R_y}{\det\hR_y}\right) + \Tr\left(\hR_y\R_y^{-1}\right) - 3M\right) + \widetilde{c}\\
=-T\cdot D_{\KL}\left(\mathcal{CN}(\uo_{3M},\hR_y),\mathcal{CN}(\uo_{3M},\R_y)\right) + \widetilde{c},
\end{gathered}
\end{equation}
where $D_{\KL}\left(\cdot,\cdot\right)$ denotes the KLD \cite{mackay2003information}, and $\widetilde{c}$ is a constant independent of $\uvarphi$. Clearly, it follows from \eqref{loglikelihoodisKLD} that maximizing the log-likelihood $\mathcal{L}(\uvarphi)$ is equivalent to minimizing the KLD $D_{\KL}\left(\mathcal{CN}(\uo_{3M},\hR_y),\mathcal{CN}(\uo_{3M},\R_y)\right)$, denoted from here on by $D^{\mathcal{CN}}_{\KL}\hspace{-0.075cm}\left(\hR_y,\R_y\right)$ for shorthand. However, since the CN distribution is fully characterized by Second-Order Statistics (SOSs), note that the KLD of two zero-mean multivariate CN distributions serves in itself as a plausible criterion for consistent covariance matrix estimation. Indeed, from Gibbs' inequality, the KLD is always non-negative (e.g., \cite{mackay2003information}), and 
\begin{equation}\label{KLDifffcond}
D^{\mathcal{CN}}_{\KL}\hspace{-0.075cm}\left(\hR_y,\R_y\right)=0 \iff \hR_y=\R_y,
\end{equation}
{\myfontb\emph{regardless}} of the true underlying distributions governing $\hR_y$ and $\R_y$. Thus, asymptotically, we now have
\begin{equation}\label{AsymptoticKLDviaOWLS}
\begin{gathered}
\hspace{-0.3cm}\huvarphi_{\KLD}\triangleq\underset{\text{{\boldmath $\varphi$}$\in\Rset^{K_{\varphi}\times1}$}}{\argmin} D^{\mathcal{CN}}_{\KL}\hspace{-0.075cm}\left(\hR_y,\R_y\right)\underset{\substack{\text{Due to} \\ \eqref{loglikelihoodisKLD}}}{=}\underset{\text{{\boldmath $\varphi$}$\in\Rset^{K_{\varphi}\times1}$}}{\argmax} \mathcal{L}(\uvarphi)\\
\;\;\;\quad\underset{\substack{\text{Due to} \\ \eqref{OWNLLSisMLE}}}{\approx}\underset{\text{{\boldmath $\varphi$}$\in\Rset^{K_{\varphi}\times1}$}}{\argmin} \left[\uvarep^{\dagger}\quad\uvarep^{\tps}\right]\R_{\varepsilon}^{-1}\left[\uvarep^{\tps}\quad\uvarep^{\dagger}\right]^{\tps}=\huvarphi_{\OWNLLS},
\end{gathered}
\end{equation}
where the optimal weight matrix $\R_{\varepsilon}^{-1}$ in \eqref{AsymptoticKLDviaOWLS}, now computed for a general ({\myfontb\emph{not necessarily}} Gaussian) signal model, is fully specified by (see Appendix \ref{AppE})
\begin{align}
\Eset\left[\ep_{ij}\ep_{k\ell}^*\right]&=\frac{1}{T}\left(\kappa_y[i,j,\ell,k]+R_{y_{ik}}R^*_{y_{j\ell}}\right),\label{covandpseudocovnonGaussian1}\\
\Eset\left[\ep_{ij}\ep_{k\ell}\right]&=\frac{1}{T}\left(\kappa_y[i,j,k,\ell]+R_{y_{i\ell}}R^*_{y_{jk}}\right),\label{covandpseudocovnonGaussian2}
\end{align}
for all $i,j,k,\ell\in\{1,\ldots,M\}$. Notice that whenever $\kappa_y[i,j,\ell,k]=0$, e.g., for CN signals, \eqref{covandpseudocovnonGaussian1}--\eqref{covandpseudocovnonGaussian2} coincide with \eqref{covandpseudocovGaussian}. Otherwise, \eqref{covandpseudocovGaussian}, which depend only on the elements of $\R_y$, can be regarded as a SOS-based approximation for \eqref{covandpseudocovnonGaussian1}--\eqref{covandpseudocovnonGaussian2}. Therefore, it follows from \eqref{AsymptoticKLDviaOWLS} that in the general case, namely regardless of the distribution of $\uy[t]$ (under the mild condition \eqref{finite4thordercumulant}), the iterations \eqref{FSAupdatequation} serve as a quasi-Newton algorithm (e.g., \cite{bonnans2006numerical}) for an implicit computation of our proposed KLD-based estimate, which is also the SOS-based approximate OWNLLS estimate. That is, the weight matrix $\R^{-1}_{\varepsilon}$ in \eqref{AsymptoticKLDviaOWLS} for this approximate OWNLLS estimate is specified in \eqref{covandpseudocovGaussian}, which depend on the SOS $\R_y$, and $T$. Of course, \eqref{AsymptoticKLDviaOWLS} also implies that $\huvarphi_{\KLD}$ is a consistent estimate.

Summarizing the above, our proposed blind DOAs estimate $\hutheta_{\KLD}$ is extracted from $\huvarphi_{\KLD}$ (as its $D$ penultimate entries, according to \eqref{defofvarphi}), which is computed via the iterations \eqref{FSAupdatequation}---the FSA for CN signals. With this, the instrumental role of the CPD-based estimate is now revealed. As seen from \eqref{AsymptoticKLDviaOWLS}, our desired KLD-based DOAs estimate is the solution of a non-convex high-dimensional optimization problem. In contrast, although the CPD-based estimate is statistically inferior, it is nevertheless consistent, and can be obtained (relatively) efficiently via AJD. Thus, the KLD-based estimate and its iterative solution \eqref{FSAupdatequation} are of practical value only when provided with a sufficiently ``good" initial estimate, which lies in the basin of attraction of the global minimizer of $D^{\mathcal{CN}}_{\KL}\hspace{-0.075cm}\left(\hR_y,\R_y\right)$. Hence the key role of $\hutheta_{\cpd},\hA_{\cpd}$ in our overall proposed solution.
\vspace{-0.2cm}
\section{Simulation Results}\label{sec:simulationresults}
In this section, we consider three simulation experiments of blind DOAs estimation for different scenarios in order to corroborate our analytical derivations by empirical results. First, for a ULA, we demonstrate the asymptotic optimality of the proposed estimate for Gaussian signals. Second, we demonstrate its consistency for non-Gaussian signals, this time with an uncalibrated ULA, exemplifying the robustness of our blind approach w.r.t.\ the signals' distributions. Lastly, we demonstrate the robustness of our blind estimate w.r.t.\ the array geometry and reliability by considering a Uniform Circular Array (UCA) with faulty sensors. In all three simulation experiments we also compare our proposed method to Zhang \etal's TriD blind estimate \cite{zhang2012trilinear}, thus demonstrating the considerable accuracy enhancement w.r.t.\ a state-of-the-art competitor.

Throughout, we consider model \eqref{modelequation} with unit variance sources, and our proposed estimates were obtained as follows:

\tcbset{colframe=gray!95!blue,size=small,width=0.49\textwidth,arc=2.1mm,outer arc=1mm}
\begin{tcolorbox}[upperbox=visible,colback=white]
\begin{enumerate}
\item Compute $\hR_x=\hR_y-\widehat{\sigma}_v^2\I_{3M}$ and construct $\htenR_x$;
\item \underline{Phase 1}: Compute $\hutheta_{\cpd},\hA_{\cpd}$ via the modified AC-DC algorithm (Section \ref{sec:ModACDCalg}), with the initial estimates $\hutheta_{\EJD},\hA_{\EJD}$ (Subsection \ref{subsec:initModACDC}); and
\item \underline{Phase 2}: Compute $\hutheta_{\KLD},\hA_{\KLD}$ via \eqref{FSAupdatequation} (the FSA for CN signals), with the initial estimates $\hutheta_{\cpd},\hA_{\cpd}$.
\end{enumerate}
\end{tcolorbox}
\noindent All empirical results were obtained by $10^4$ independent trials.

In the first two experiments which follow (Subsections \ref{subsec:GaussianSignals} and \ref{subsec:QPSKLaplace}), we consider a ULA with $M=7$ AVSs and half wavelength inter-element spacing (i.e., $\lambda/2$), and $D=3$ sources impinging from azimuth angles $\utheta=\left[-56^\circ\;43^\circ\;71^\circ\right]^{\tps}$.
\begin{figure*}[]
	\centering
	\begin{subfigure}[b]{0.49\textwidth}
		\includegraphics[width=\textwidth]{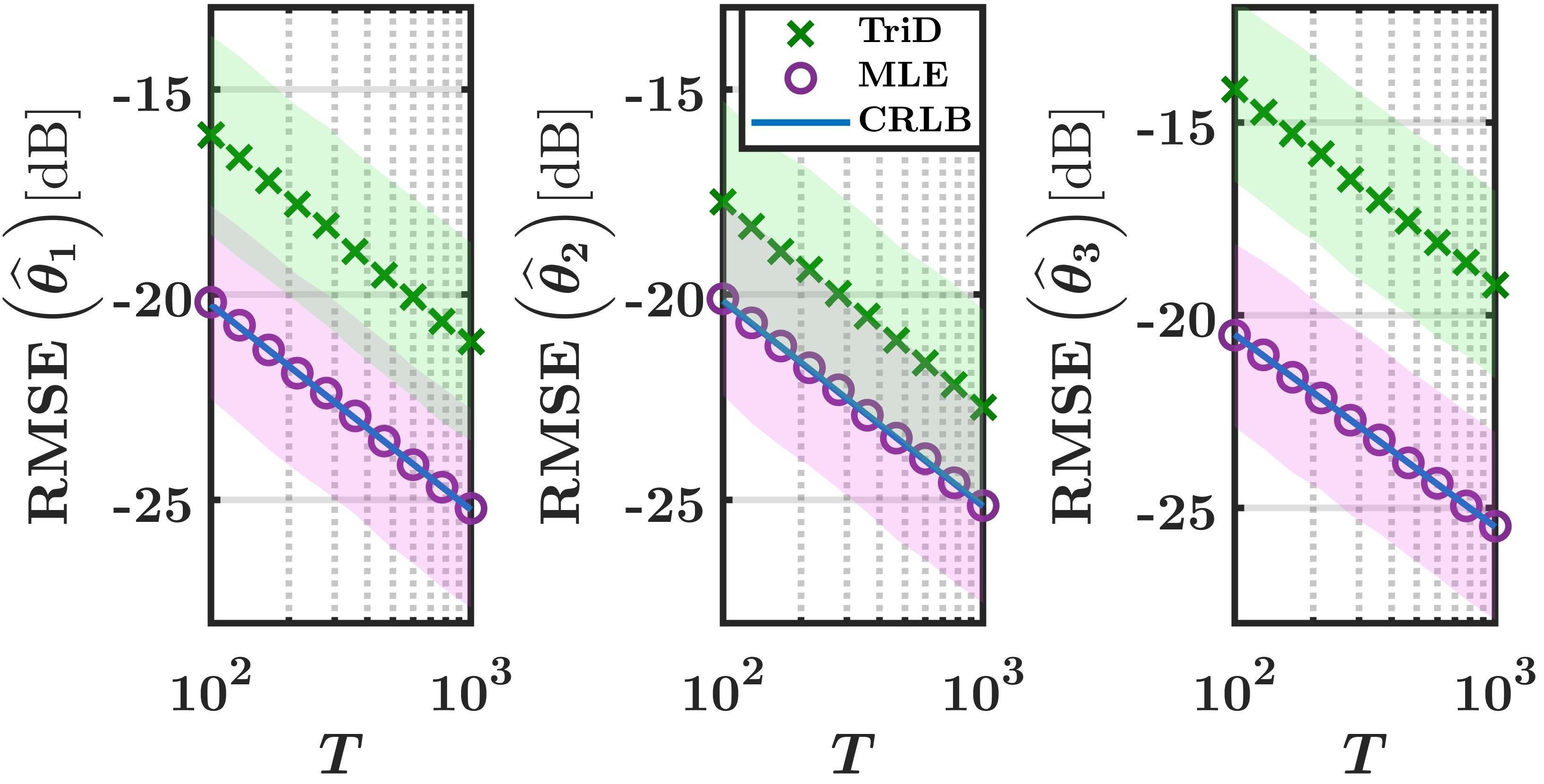}\vspace{-0.1cm}
		\caption{}\vspace{-0.1cm}
		\label{fig:exp1_vs_T}
	\end{subfigure}%
	~
	\begin{subfigure}[b]{0.49\textwidth}
		\includegraphics[width=\textwidth]{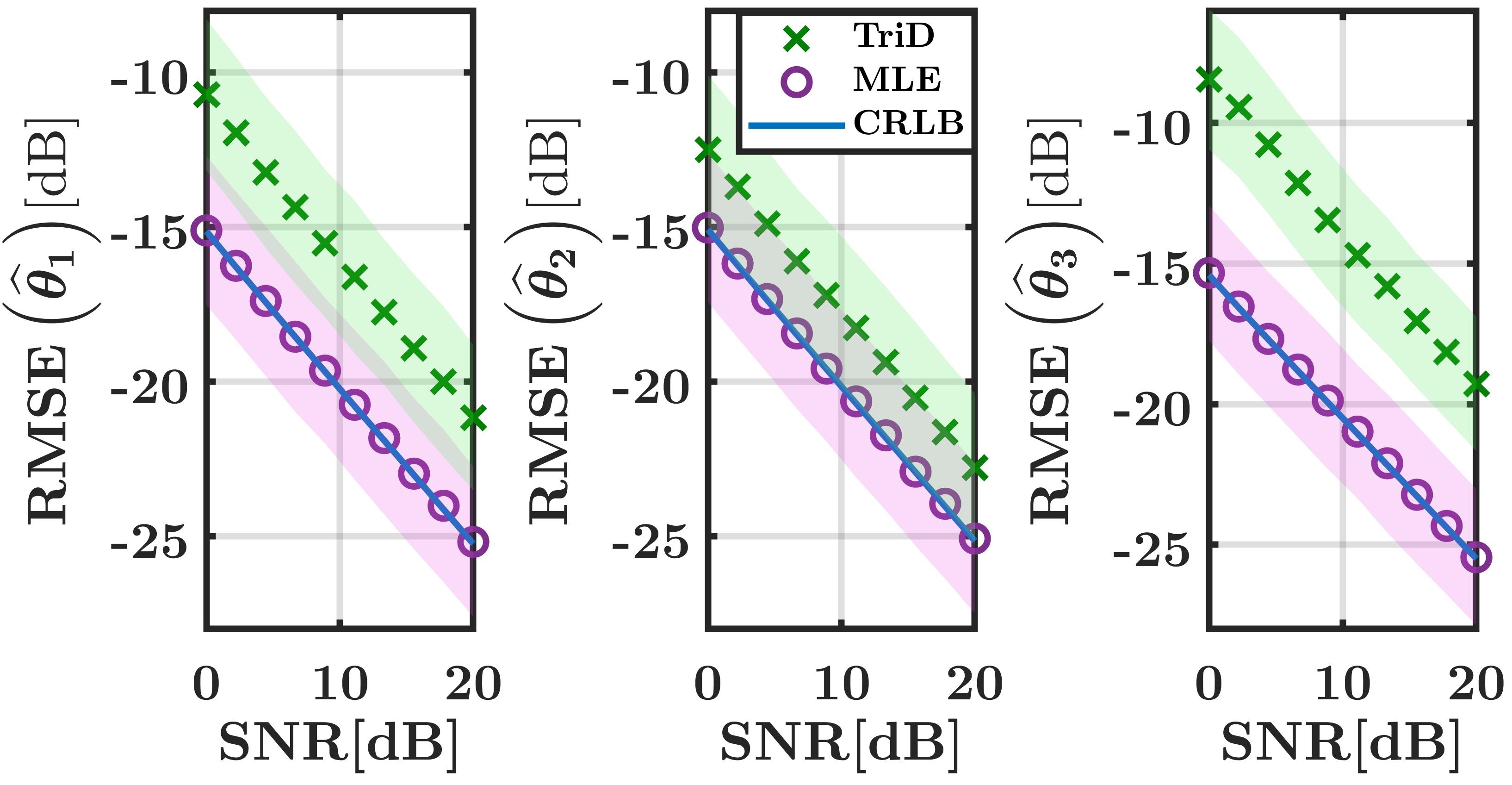}\vspace{-0.1cm}
		\caption{}\vspace{-0.1cm}
		\label{fig:exp1_vs_SNR}
	\end{subfigure}
	\caption{RMSE of the DOAs estimates with standard deviations envelopes. (a) vs.\ $T$, for a fixed SNR level of $10$[dB] (b) vs.\ SNR, for a fixed sample size of $T=100$. Evidently, for CN signals our proposed estimates are optimal and attain the CRLB. The gain w.r.t.\ TriD in this scenario reaches up to $\sim\hspace{-0.05cm}5$[dB].}
	\label{fig:exp1}
\end{figure*}
\begin{figure*}[]
	\centering
	\includegraphics[width=\textwidth]{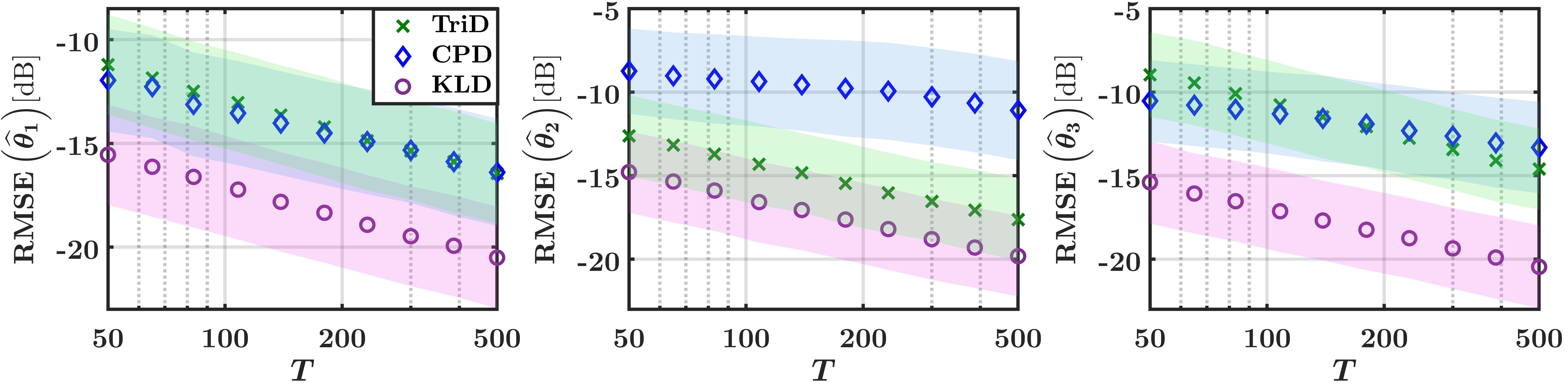}
	\caption{RMSE of the DOAs estimates vs.\ $T$, for a fixed SNR level of $5$[dB], with standard deviations envelopes. While the CPD-based (phase 1) estimates are only competitive to the TriD estimates, the proposed enhanced KDL-based (phase 2) estimates exhibit a considerable performance gain w.r.t.\ TriD for this scenario as well (QPKS source and Laplace noise), reaching up to $\sim\hspace{-0.05cm}7$[dB].}\vspace{-0.1cm}
	\label{fig:exp2}
\end{figure*}
\begin{figure*}[]
	\centering
	\begin{subfigure}[b]{0.49\textwidth}
		\includegraphics[width=\textwidth]{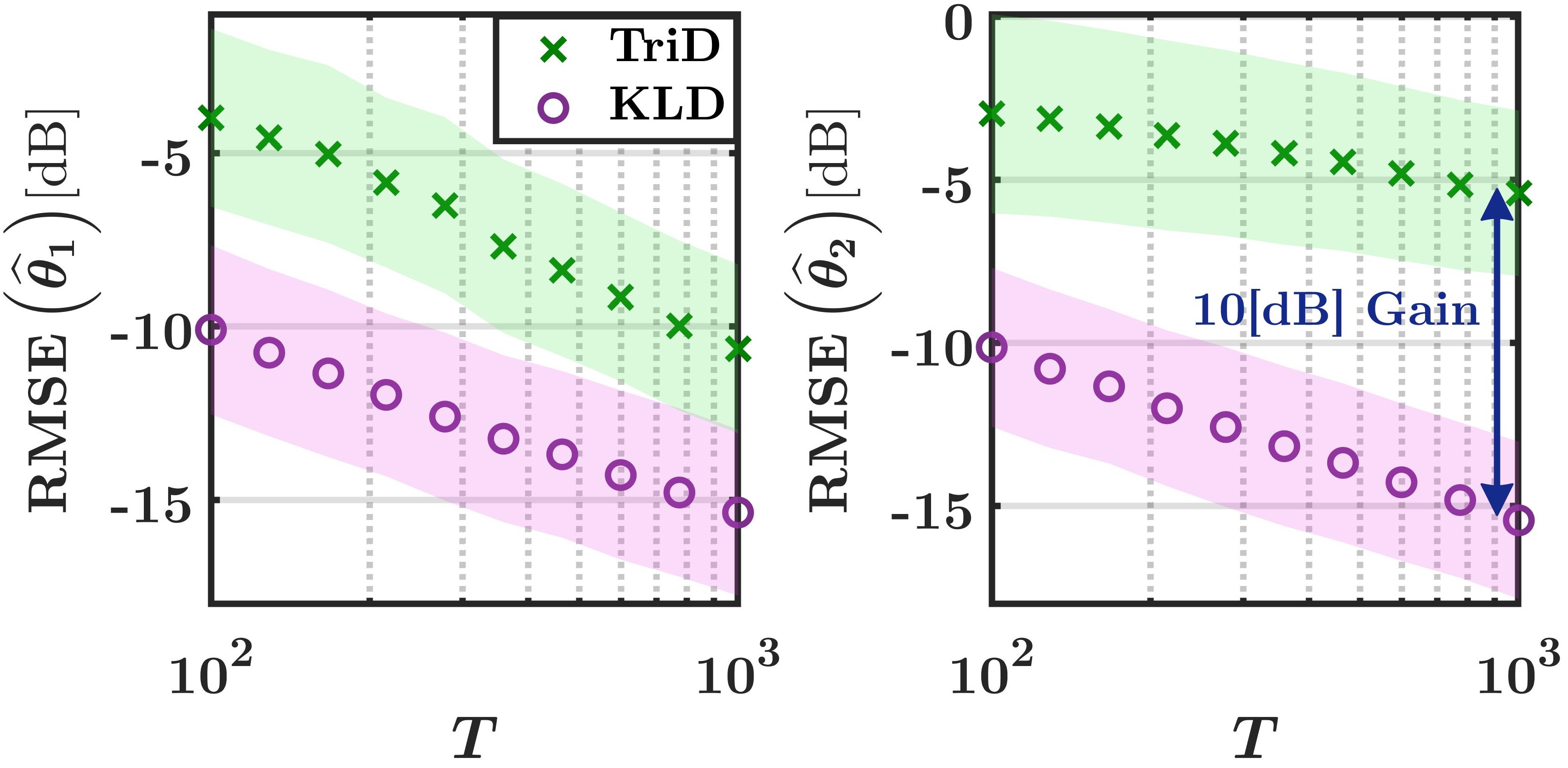}\vspace{-0.1cm}
		\caption{}\vspace{-0.1cm}
		\label{fig:exp3_vs_T}
	\end{subfigure}%
	~
	\begin{subfigure}[b]{0.49\textwidth}
		\includegraphics[width=\textwidth]{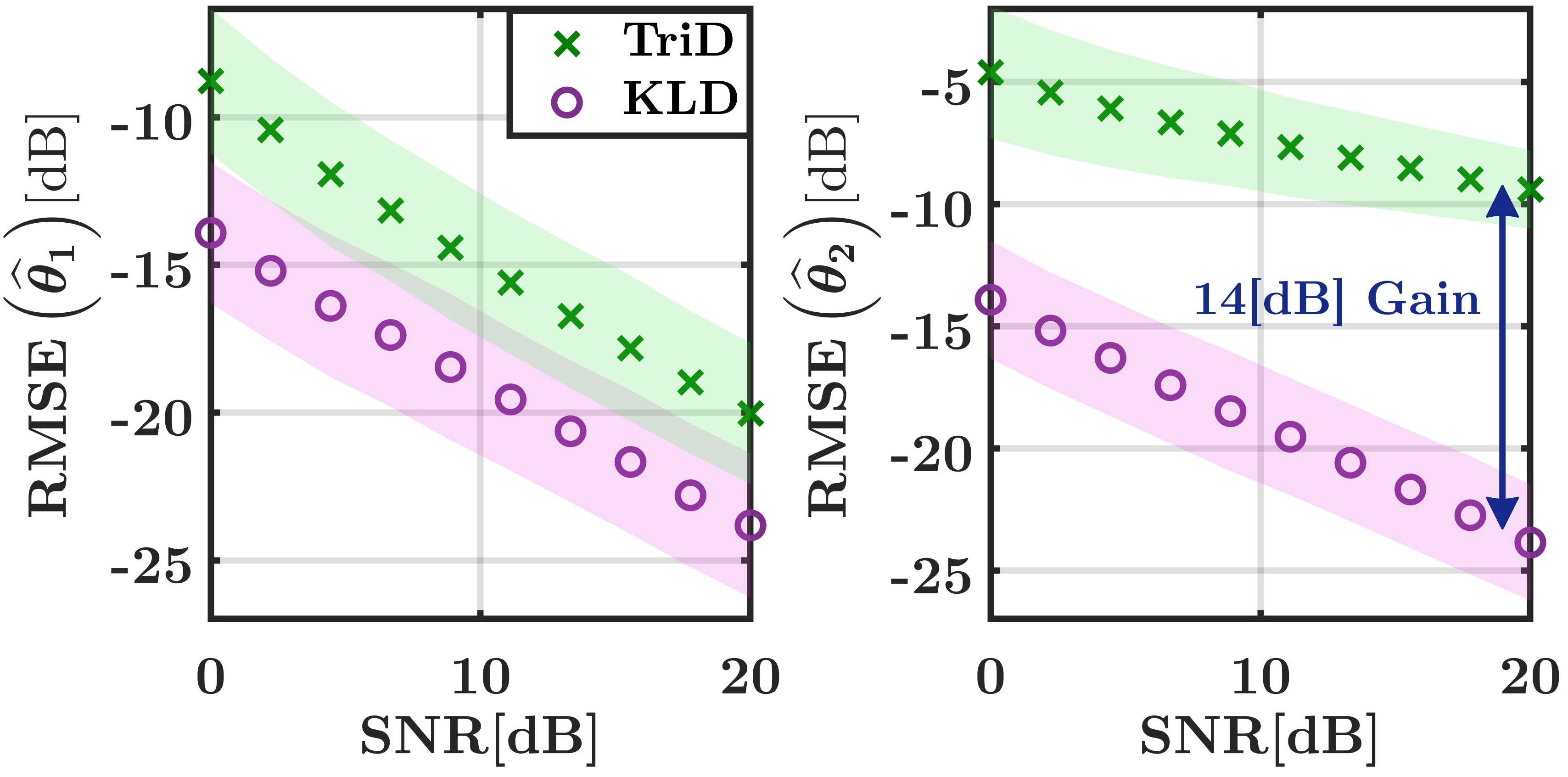}\vspace{-0.1cm}
		\caption{}\vspace{-0.1cm}
		\label{fig:exp3_vs_SNR}
	\end{subfigure}
	\caption{RMSE of the DOAs estimates with standard deviations envelopes. (a) vs.\ $T$, for a fixed SNR level of $0$[dB] (b) vs.\ SNR, for a fixed sample size of $T=500$. Here as well, for a UCA with faulty elements and Gaussian mixture sources, the improved accuracy is substantial, reaching up to $\sim14$[dB].}
	\label{fig:exp3}\vspace{-0.4cm}
\end{figure*}
\begin{figure}[]
	\centering
	\includegraphics[width=0.49\textwidth]{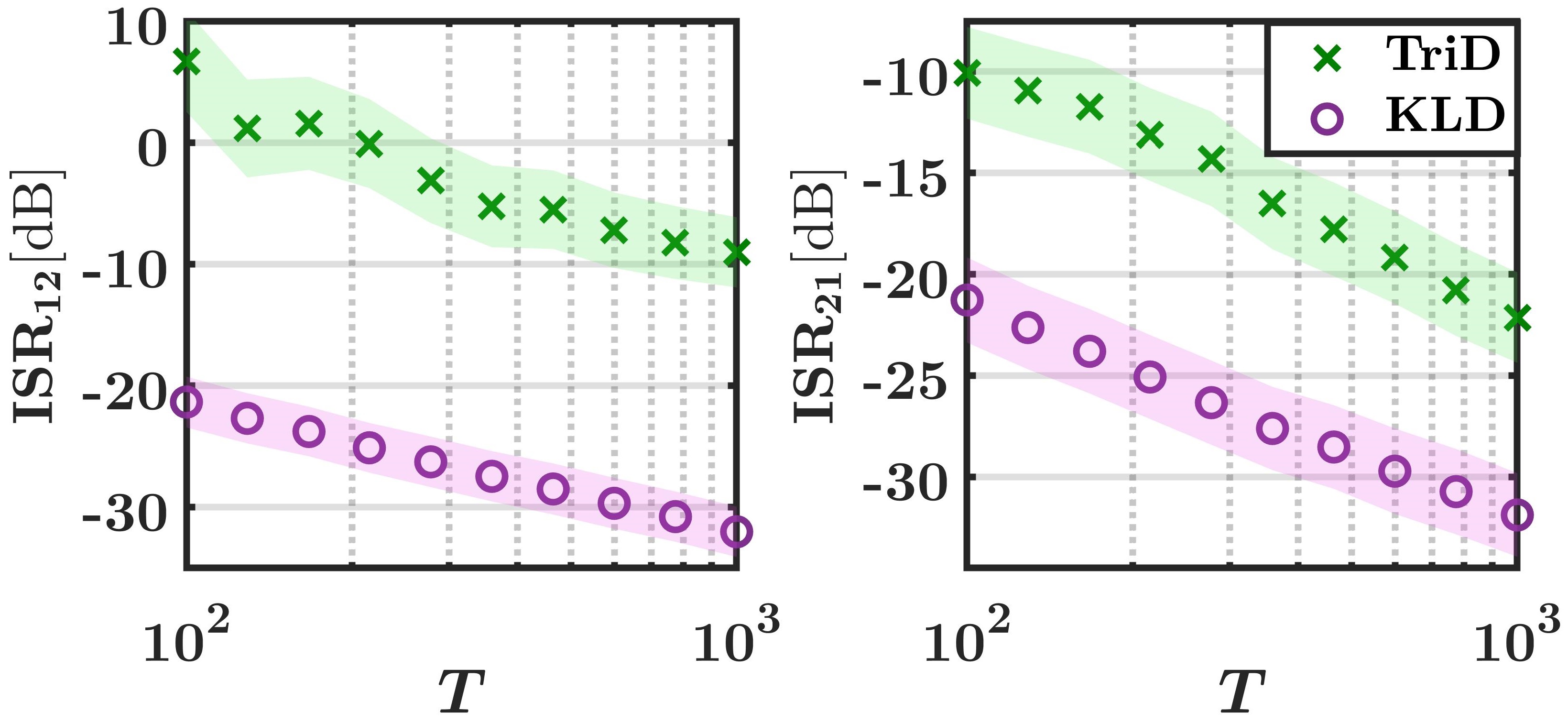}
	\caption{ISR vs.\ $T$, for a fixed SNR level of $0$[dB], with standard deviations envelopes. On top of higher accuracy in DOAs estimation, our KLD-based estimate of the mixing matrix $\A$ also yields enhanced separation performance relative the TridD estimate.}\vspace{-0.1cm}
	\label{fig:exp4}
\end{figure}
\vspace{-0.7cm}
\subsection{Optimal Performance for Gaussian Signals}\label{subsec:GaussianSignals}
In this experiment CN sources are considered. Therefore, the received signals' distribution is prescribed by \eqref{CN_samples}, and it follows that $\hutheta_{\KLD}=\hutheta_{\ML}$ (and also $\hA_{\KLD}=\hA_{\ML}$). Fig.\ \ref{fig:exp1_vs_T} presents the Root MSE (RMSE) of the DOAs estimates vs.\ the sample size $T$ for a fixed SNR level of $10$[dB]. Likewise, Fig.\ \ref{fig:exp1_vs_SNR} presents the RMSE of the DOAs estimates, this time vs.\ the SNR for a fixed sample size of $T=100$. Clearly, it is seen that our proposed estimates attain the CRLB, namely their optimality is demonstrated. Further, the accuracy improvement w.r.t.\ the TriD method (up to $\sim\hspace{-0.05cm}5$[dB]) is uniformly obtained, for all the DOAs in any sample size and/or any SNR.
\vspace{-0.3cm}
\subsection{QPSK Sources, Laplace Noise and an Uncalibrated Array}\label{subsec:QPSKLaplace}
In the second experiment we consider the same ULA as in the first one, but now with inter-AVSs' positioning errors and gain offsets \cite{song2014quasi}. Formally, each element of the {\myfontb\emph{uncalibrated}} steering vectors matrix $\A$ now reads
\begin{equation*}\label{manifoldmatrixoffsets}
\begin{gathered}
A_{md}=A_{md}'\cdot g_m\cdot e^{\jmath k\left(\cos(\theta_d)\Delta_x^{(m)}+\sin(\theta_d)\Delta_y^{(m)}\right)},\\
\forall m\in\{1,\ldots,7\},\; \forall d\in\{1,2,3\}
\end{gathered}
\end{equation*}
where $g_m\in\Rset_+$ and $\Delta_x^{(m)},\Delta_y^{(m)}\in\Rset$ are the gain offset and the $x$- and $y$-axis positioning errors of the $m$-th AVS, resp., $k\triangleq2\pi/\lambda$ is the sources' wavenumber, and $A_{md}'\triangleq\exp\left(\jmath\pi(m-1)\cos(\theta_d)\right)$. In our simulation, the gain offsets $\{g_m\}$ and the positioning errors, $\{\Delta_x^{(m)}\}$ and $\{\Delta_y^{(m)}\}$, were independently drawn (once, and then fixed) from the uniform distributions $U(0.7,1.3)$ and $U(-1,1)$, resp. In order to demonstrate the robustness of our method w.r.t.\ the signals' distributions, we consider Quadrature Phase-Shift Keying (QPSK) sources. Moreover, here the real and imaginary parts of the noise $\uv[t]$ are independent Laplace distributed, intentionally \emph{not} in accordance with our CN assumption, in order to demonstrate our proposed estimate's robustness to the noise distribution. Fig.\ \ref{fig:exp2} presents the DOAs RMSE vs.\ the sample size $T$ for a fixed SNR level of $5$[dB], and implicitly demonstrate that \eqref{AsymptoticKLDviaOWLS} holds, and hence the preceding derivation, particularly \eqref{OWNLLSisMLE}. Indeed, here as well both blind estimates are robust w.r.t.\ the gain and phase offsets, though similarly to the previous scenario examined in Subsection \ref{subsec:GaussianSignals}, a considerable improvement is still obtained relative to the TriD method up to $\sim\hspace{-0.1cm}7$[dB]. Additionally, in this experiment only, we also present the performance of the (phase 1) CPD-based estimates, which play the key role of initial estimates in the computation of the proposed enhanced KLD-based estimates. Similar trends as in Fig.\ \ref{fig:exp1_vs_SNR} were obtained in simulations for different SNR levels in this scenario as well.
\vspace{-0.3cm}
\subsection{Gaussian Mixtures and a UCA with Faulty Elements}\label{subsec:UCAFaultyelements}
In our last experiment, we consider a UCA with $M=5$ equiangular spaced AVSs placed on its circumference \cite{ioannides2005uniform}. In addition to changing the array geometry, we further investigate the performance of the proposed method in the presence of sensors failures (e.g., \cite{vigneshwaran2007direction}). Specifically, we assume that the $2^{\text{nd}}$ and $4^{\text{th}}$ AVSs are faulty, such that all their (three) elements do not receive any external signal, and thus contain only noise. Assuming for simplicity that the radius of the UCA is $\lambda/2$,
\begin{equation*}\label{manifoldmatrixfaultyelements}
A_{md}=\begin{cases}
e^{\jmath\pi\cos\left(\theta_d-\tfrac{2\pi(m-1)}{M}\right)}, & m=1,3,5\\
0, & m=2,4
\end{cases},\,\forall d\in\{1,2\}.
\end{equation*}
Further, we consider $D=2$ sources impinging from angles $\utheta=\left[24^\circ\;92^\circ\right]^{\tps}$, which are complex-valued Gaussian mixtures (i.e., both real and imaginary parts are i.i.d. Gaussian mixtures) of order $2$, with means $+\tfrac{1}{\sqrt{2}},-\tfrac{1}{\sqrt{2}}$ and variances $\tfrac{1}{2},\tfrac{1}{2}$ for the first and second equiprobable Gaussian components, resp.

As evident from Figs.\ \ref{fig:exp3_vs_T} and \ref{fig:exp3_vs_SNR}, presenting the DOAs RMSEs vs.\ the sample size (for fixed SNR of $0$[dB]) and the SNR (for fixed $T=500$), resp., similar trends are obtained as in the previous two experiments. This not only demonstrates the robustness of our proposed approach w.r.t.\ the array geometry, sensors malfunctioning, and the sources' distribution, but also indicates that the substantial accuracy improvement, which in this case reaches up to more than an order of magnitude, is obtained in various different scenarios. 

Lastly, we consider the Interference-to-Source Ratio (ISR),
\begin{equation}\label{ISRdefinition}
\text{ISR}_{ij}\triangleq\Eset\left[\frac{|\hA^+\A|^2_{ij}}{|\hA^+\A|^2_{ii}}\right],\; \forall i\neq j\in\{1,\ldots,D\},
\end{equation}
which measures the residual energy of the $j$-th source in the reconstruction of the $i$-th source, and is a common measure for the separation performance. Fig.\ \ref{fig:exp4}, presenting the ISR vs.\ the sample size, shows that our estimate for the mixing matrix $\A$ also yields better source separation than the TriD estimate.
\vspace{-0.3cm}
\section{Conclusion}\label{sec:conclusion}
In the context of passive AVS arrays, we presented a novel blind DOAs estimate. Rather than estimating the DOAs directly from the raw data, our estimate, obtained in a two-phase procedure, exploits the special structure of the observations' covariance matrix. In the first phase, the CPD-based estimates are obtained via the modified AC-DC algorithm, based on the the unique quadrilinear decomposition of the SOSs covariance tensor. In the second phase, the proposed estimates, which were shown to be asymptotically equivalent the OWNLLS estimates, are obtained via KLD covariance fitting. Since minimization of the KLD is equivalent to maximization of the Gaussian ML objective, our proposed KLD-based estimates are optimal for Gaussian signals, and can be computed via the FSA for non-Gaussian signals as well. Our analytical results were supported by simulation experiments in various scenarios, which demonstrated our estimate's robustness, as well as its superiority over another leading blind DOAs estimate.
\vspace{-0.35cm}
\appendices
\section{Proof of Theorem \ref{covtensoruniqeCPD}}\label{AppA}
\begin{proof}
In order to prove Theorem \ref{covtensoruniqeCPD}, we shall invoke Theorem \ref{tensoruniqueCPD}.
To this end, note that $\tenR_x$, given explicitly in \eqref{covtensorrankD2}, may also be written in the form \eqref{fourmodetensorthm} of Theorem \ref{tensoruniqueCPD} as
\begin{equation}\label{fourmodetensortenR_x}
\begin{gathered}
\mathcal{R}_x(i,j,m,n)=\sum_{d=1}^{D}{C_{id}(\utheta)C_{jd}(\utheta)A_{md}A^*_{nd}}\in\Cset,\\
\forall i,j\in\{1,2,3\},\;\forall m,n\in\{1,\ldots,M\},
\end{gathered}
\end{equation}
with $\C(\utheta),\C(\utheta),\A$ and $\A^*$ as the four factor matrices of $\tenR_x$. Then, according to Theorem \ref{tensoruniqueCPD}, the factor matrices, which are in this case $\C(\utheta)$ and $\A$, are unique up to permutation and complex scaling of columns provided that
\begin{equation}\label{kruskalcondAppA}
k_{\text{\boldmath $C(\utheta)$}}+k_{\text{\boldmath $C(\utheta)$}}+k_{\text{\boldmath $A$}}+k_{\text{\boldmath $A$}^*}\geq 2D+3.
\end{equation}
From the sensor array regularity condition, we have $\text{rank}(\A)=D$, which also implies $\text{rank}(\A^*)=\text{rank}(\A)=D$. Since $\A$ (and therefore $\A^*$) is full column rank, it follows that $k_{\text{\boldmath $A$}}=k_{\text{\boldmath $A$}^*}=D$. In addition, since $\theta_d\neq\theta_\ell$ for all $d\neq\ell$ by assumption, it follows that $\C(\utheta)$ is also full column rank, which means that $k_{\text{\boldmath $C(\utheta)$}}=\text{rank}(\C(\utheta))=\min\{3,D\}$. Therefore, in our case, the uniqueness condition \eqref{kruskalcondAppA} reads
\begin{equation}\label{kruskalcondAppAspecific}
2\cdot\min\{3,D\}+2D\geq 2D+3 \;\Longrightarrow\; 2\cdot\min\{3,D\}\geq 3.
\end{equation}
For $D\geq3$, \eqref{kruskalcondAppAspecific} obviously hold. For the complementary case $D<3$, \eqref{kruskalcondAppAspecific} becomes $2D\geq3$. Therefore, we conclude that for $D>1$, which holds by assumption, the CPD \eqref{fourmodetensortenR_x} is unique up to permutation and complex scaling of columns. Finally, since the first row of $\C(\utheta)$ is the all-ones $D$-dimensional vector and $\{A_{1d}\in\Rset_{\geq0}\}$, thus eliminating the complex scaling ambiguity, and since $\theta_1<\ldots<\theta_D$, thus eliminating the permutation ambiguity, we conclude that \eqref{fourmodetensortenR_x} is unique.
\end{proof}
\section{Proof of Theorem \ref{consistencyofCPDestimate}}\label{AppB}
\begin{proof}
To prove Theorem \ref{consistencyofCPDestimate}, we shall use the basic consistency theorem for {\myfontb\emph{extremum estimators}} \cite{newey1994large}. To this end, for shorthand, let us first define the vector of real-valued unknowns $\tuvartheta\triangleq\left[\text{vec}(\Re\{\tA\})^{\tps}\,\text{vec}(\Im\{\widetilde{\I}_M\tA\})^{\tps}\,\tutheta^{\tps}\right]^{\tps}\in\Rset^{2MD\times 1}$, where $\widetilde{\I}_M\triangleq\left[\uo_{M-1}\;\ue_{1}\;\cdots\;\ue_{M}\right]\in\Rset^{(M-1)\times M}$ preserves all the rows except for the first one of the matrix it is left-multiplied with, and further define $\widehat{Q}_T(\tuvartheta)\triangleq-\norm{\tenR\left(\tutheta,\tA\right)-\htenR_x}^2_{\rm{F}}$. Then, by virtue of Theorem 2.1 in \cite{newey1994large} due to Newey and McFadden, if there exist a function $Q_0(\tuvartheta)$ such that:
\begin{enumerate}[(i)]
	\item $Q_0(\tuvartheta)$ is uniquely maximized at $\uvartheta$;
	\item $\tuvartheta\in\Theta$, where $\Theta$ is a compact set;
	\item $Q_0(\tuvartheta)$ is continuous; and
	\item $\widehat{Q}_T(\tuvartheta)$ converges uniformly in probability to $Q_0(\tuvartheta)$,
\end{enumerate}
then 
\begin{equation}\label{consistencyAppB}
\underset{\widetilde{\text{\boldmath$\vartheta$}}\in\Theta}{\argmax}\,\widehat{Q}_T(\tuvartheta)\xrightarrow[\quad\;]{p}\uvartheta\,\Longrightarrow\,\left(\hutheta_{{\cpd}},\hA_{{\cpd}}\right)\xrightarrow[\quad\;]{p}\left(\utheta,\A\right),
\end{equation}
namely \eqref{CPDDOAestimate} are consistent estimates.

We shall now show that the four conditions above (i)--(iv) hold for the function $Q_0(\tuvartheta)\triangleq-\norm{\tenR\left(\tutheta,\tA\right)-\tenR_x}^2_{\rm{F}}$. First, condition (i) is an immediate consequence of Theorem \ref{covtensoruniqeCPD}, the uniqueness of the CPD of $\tenR_x$, proved in Appendix \ref{AppA}. Second, conditions \eqref{boundingspherecondition} and \eqref{DOAscompactsetcondition} stated in the Theorem imply that both $\tutheta$ and $\tA$ belong to compact sets. The union of these compact sets, denoted $\Theta$ henceforth, is a compact set as well, hence $\tuvartheta$ belong to the compact set $\Theta$, and condition (ii) is satisfied. Third, notice that $\tenR\left(\tutheta,\tA\right)$ is a continuous function w.r.t.\ $\tutheta$ and $\tA$ by its definition \eqref{parametrictensorfunction}, in the sense that each element of $\tenR\left(\tutheta,\tA\right)$ is a continuous function w.r.t.\ each of the elements of $\tutheta$ and of $\tA$. Additionally, note that $\widetilde{Q}_0(\tenX)\triangleq-\norm{\tenX-\tenR_x}^2_{\rm{F}}$ is a continuous function w.r.t.\ $\tenX$. Now, since $Q_0(\tuvartheta)=\widetilde{Q}_0\left((\tenR\left(\tutheta,\tA\right)\right)$, the function $Q_0(\tuvartheta)$ is a composition of continuous functions, and is therefore continuous, thus  (iii) holds. To show that condition (iv) holds, i.e., uniform convergence in probability \cite{newey1991uniform}, we first present the notion of {\myfontb\emph{stochastic equicontinuity}} (\cite{newey1991uniform}, Section 2):

\noindent\textit{$\widehat{Q}_T(\uvartheta)$ is said to be stochastically equicontinuous if for every $\epsilon,\eta>0$ there exist a sequence of random variables $\widehat{\Delta}_T(\epsilon,\eta)$ and a samples size $T_0(\epsilon,\eta)$ such that for $T\geq T_0(\epsilon,\eta)$, $\Pr\left(\left|\widehat{\Delta}_T(\epsilon,\eta)\right|>\epsilon\right)<\eta$, and for each $\uvartheta$ there is an open set $U_{\text{\boldmath$\vartheta$}}$ containing $\uvartheta$ with
	\begin{equation*}\label{stochasticequicon}
		\sup_{\text{\boldmath$\widetilde{\vartheta}$}\in U_{\text{\boldmath$\vartheta$}}}\left|\widehat{Q}_T(\tuvartheta)-\widehat{Q}_T(\uvartheta)\right|\leq\widehat{\Delta}_T(\epsilon,\eta),\quad T\geq T_0(\epsilon,\eta).
\end{equation*}}

Now, according to Lemma 2.8 in \cite{newey1994large}, if $\tuvartheta$ belong to a compact set and $Q_0(\tuvartheta)$ is continuous, then condition (iv) holds \textit{if and only if}
\begin{enumerate}[(a)]
	\item $\widehat{Q}_T(\tuvartheta)\xrightarrow[\quad\;]{p}Q_0(\tuvartheta)$ for all $\tuvartheta\in\Theta$; and
	\item $\widehat{Q}_T(\tuvartheta)$ is stochastically equicontinuous.
\end{enumerate}
As we have already shown above, $\tuvartheta$ belong to the compact set $\Theta$ and $Q_0(\tuvartheta)$ is continuous. Thus, if (a) and (b) hold, condition (iv) is fulfilled, and the proof is completed.

By the assumption in Theorem \ref{consistencyofCPDestimate}, $\hR_y$ and $\widehat{\sigma}_v^2$ are consistent. Hence $\hR_x$ defined in \eqref{covmatofxest} is consistent, and also $\htenR_x$. Further, note that $\widehat{Q}_T(\tuvartheta)$ is continuous w.r.t.\ $\htenR_x$. Therefore, by the Continuous Mapping Theorem (CMT) \cite{mann1943stochastic}, for all $\tuvartheta\in\Theta$,
\begin{equation*}
	\htenR_x\xrightarrow[\quad\;]{p}\tenR_x\;\Longrightarrow\;\widehat{Q}_T(\tuvartheta)\xrightarrow[\quad\;]{p}Q_0(\tuvartheta),
\end{equation*}
thus (a) holds. We now turn to the final phase, proving (b).

For brevity, let $\R\left(\tuvartheta\right)\coloneq\R\left(\tutheta,\tA\right)$ in accordance with \eqref{blockmatrixRx}. Using the continuity of the Frobenius norm and of $\R\left(\tuvartheta\right)$, for all $\epsilon_1,\epsilon_2>0$, there exist $\rho_1,\rho_2>0$ such that
\begin{align}\label{conditionepsilons1}
	&\norm{\tuvartheta-\uvartheta}_2<\rho_1\;\Longrightarrow\;\left|\norm{\R\left(\tuvartheta\right)}^2_{\rm{F}}-\norm{\R(\uvartheta)}^2_{\rm{F}}\right|<\epsilon_1,\nonumber\\
	&\norm{\tuvartheta-\uvartheta}_2<\rho_2\;\Longrightarrow\;\norm{\R\left(\tuvartheta\right)-\R\left(\uvartheta\right)}_{\rm{F}}<\epsilon_2.
\end{align}
Therefore, $\norm{\tuvartheta-\uvartheta}_2<\rho_{\text{min}}\triangleq\min\{\rho_1,\rho_2\}$ implies
\begin{equation}\label{minconditionAppB}
	\left|\norm{\R\left(\tuvartheta\right)}^2_{\rm{F}}-\norm{\R(\uvartheta)}^2_{\rm{F}}\right|<\epsilon_1\;\cap\;\norm{\R\left(\tuvartheta\right)-\R\left(\uvartheta\right)}_{\rm{F}}<\epsilon_2.
\end{equation}
Now, pick $\epsilon,\eta>0$, and define $\widehat{\Delta}_T(\epsilon,\eta)\triangleq\tfrac{\epsilon}{2}\cdot\tfrac{1+\norm{\widehat{\text{\boldmath$R$}}_x}_{\rm{F}}}{1+\norm{\text{\boldmath$R$}_x}_{\rm{F}}}$. From the CMT, $\hR_x\xrightarrow[\quad\;]{p}\R_x$ implies $\widehat{\Delta}_T(\epsilon,\eta)\xrightarrow[\quad\;]{p}\tfrac{\epsilon}{2}$. Therefore,
\begin{align}
	&\exists T_0(\epsilon,\eta)\in\Nset:\nonumber\\
	&\Pr\left(\left|\widehat{\Delta}_T(\epsilon,\eta)\right|>\epsilon\right)=\Pr\left(\tfrac{1+\norm{\widehat{\text{\boldmath$R$}}_x}_{\rm{F}}}{1+\norm{\text{\boldmath$R$}_x}_{\rm{F}}}>2\right)<\eta,
\end{align}
for $T\geq T_0(\epsilon,\eta)$. In addition, define $\E(\tuvartheta,\uvartheta)\triangleq\R\left(\tuvartheta\right)-\R\left(\uvartheta\right)$, $\epsilon_1\triangleq\tfrac{\epsilon}{2}\cdot\tfrac{1}{1+\norm{\text{\boldmath$R$}_x}_{\rm{F}}}$ and $\epsilon_2\triangleq\tfrac{\epsilon}{4}\cdot\tfrac{1}{1+\norm{\text{\boldmath$R$}_x}_{\rm{F}}}$, and for each $\uvartheta$ further define the open set $U_{\text{\boldmath$\vartheta$}}\triangleq\left\{\tuvartheta:\norm{\tuvartheta-\uvartheta}_2<\rho_{\text{min}}\right\}$. Then, with these notations we now have
\begin{align*}\label{laststepofproofAppB}
	&\sup_{\widetilde{\text{\boldmath$\vartheta$}}\in U_{\text{\boldmath$\vartheta$}}}\left|\widehat{Q}_T(\tuvartheta)\hspace{-0.025cm}-\hspace{-0.025cm}\widehat{Q}_T(\uvartheta)\right|\underset{\text{(I)}}{=}\nonumber\\
	&\sup_{\widetilde{\text{\boldmath$\vartheta$}}\in U_{\text{\boldmath$\vartheta$}}}\left|\norm{\R\left(\tuvartheta\right)}^2_{\rm{F}}\hspace{-0.025cm}-\hspace{-0.025cm}\norm{\R\left(\uvartheta\right)}^2_{\rm{F}}\hspace{-0.025cm}-\hspace{-0.025cm}2\hspace{0.05cm}\Tr\left(\E(\tuvartheta,\uvartheta)\hR_x^{\dagger}\right)\right|\hspace{-0.025cm}\underset{\text{(II)}}{\leq}\nonumber\\
	&\sup_{\widetilde{\text{\boldmath$\vartheta$}}\in U_{\text{\boldmath$\vartheta$}}}\hspace{-0.05cm}\left\{\left|\norm{\R\left(\tuvartheta\right)}^2_{\rm{F}}\hspace{-0.05cm}-\hspace{-0.05cm}\norm{\R\left(\uvartheta\right)}^2_{\rm{F}}\right|\hspace{-0.05cm}+\hspace{-0.05cm}2\left|\Tr\left(\E(\tuvartheta,\uvartheta)\hR_x^{\dagger}\right)\right|\right\}\hspace{-0.075cm}\underset{\text{(III)}}{\leq}\nonumber\\
	&\sup_{\widetilde{\text{\boldmath$\vartheta$}}\in U_{\text{\boldmath$\vartheta$}}}\hspace{-0.05cm}\left\{\left|\norm{\R\left(\tuvartheta\right)}^2_{\rm{F}}\hspace{-0.05cm}-\hspace{-0.05cm}\norm{\R\left(\uvartheta\right)}^2_{\rm{F}}\right|\hspace{-0.05cm}+\hspace{-0.05cm}2\norm{\E(\tuvartheta,\uvartheta)}_{\rm{F}}\norm{\hR_x}_{\rm{F}}\right\}\hspace{-0.075cm}\underset{\text{(IV)}}{\leq}\nonumber\\
	&\frac{\epsilon}{2}\cdot\frac{1}{1+\norm{\text{\boldmath$R$}_x}_{\rm{F}}}+2\cdot\frac{\epsilon}{4}\cdot\frac{1}{1+\norm{\text{\boldmath$R$}_x}_{\rm{F}}}\norm{\hR_x}_{\rm{F}}=\widehat{\Delta}_T(\epsilon,\eta),
\end{align*}
for $T\geq T_0(\epsilon,\eta)$, where we have used
\begin{enumerate}[(I)]
	\item $\widehat{Q}_T(\tuvartheta)=-\norm{\R\left(\tuvartheta\right)}^2_{\rm{F}}-\norm{\hR_x}^2_{\rm{F}}+2\hspace{0.05cm}\Tr\left(\R(\tuvartheta)\hR_x^{\dagger}\right)$;
	\item $\left|a-b\right|\leq|a|+|b|$ (triangle inequality);
	\item $\Tr\left(\hspace{-0.05cm}\A\B^{\dagger}\hspace{-0.05cm}\right)\hspace{-0.05cm}\leq\hspace{-0.05cm}\norm{\A}_{\rm{F}}\hspace{-0.05cm}\norm{\B}_{\rm{F}}$ (Cauchy-Schwarz inequality);
	\item $\tuvartheta\in U_{\text{\boldmath$\vartheta$}}\;\Longrightarrow\;\norm{\tuvartheta-\uvartheta}_2<\rho_{\text{min}}\;\Longrightarrow$\; \eqref{minconditionAppB}.
\end{enumerate}
Thus, $\widehat{Q}_T(\tuvartheta)$ is stochastically equicontinuous, i.e., (b) holds, which implies that (iv) holds. In conclusion, conditions (i)--(iv) hold, hence \eqref{consistencyAppB} follows, and $\hutheta_{{\cpd}}, \hA_{{\cpd}}$ are consistent.
\end{proof}

\section{Derivation of the DOAs Alternating LS equations}\label{AppC}
Let us begin by rewriting the cost function \eqref{costfunctionLS} as
\begin{equation}\label{LScostinTraceform}
\begin{gathered}
C_{\LS}\left(\utheta,\A\right)=\norm{\tenR\left(\utheta,\A\right)-\htenR_x}^2_{\rm{F}}=\norm{\R_x-\hR_x}^2_{\rm{F}}=\\
\Tr\left(\R_x\R_x^{\dagger}\right)-2\cdot\Tr\left(\R_x\hR_x^{\dagger}\right)+\Tr\left(\hR_x\hR_x^{\dagger}\right),
\end{gathered}
\end{equation}
where $\R_x$ depends on $\utheta$ and $\A$ as prescribed in \eqref{blockmatrixRx}. Now, differentiating \eqref{LScostinTraceform} w.r.t.\ $\theta_d$ gives
\begin{equation}\label{derivativeLS_wrt_theta}
\frac{\partial C_{\LS}}{\partial\theta_d}=\frac{\partial}{\partial\theta_d}\Tr\left(\R_x\R_x\right)-2\cdot\frac{\partial}{\partial\theta_d}\Tr\left(\R_x\hR_x\right),
\end{equation}
where have have used $\R_x=\R_x^{\dagger}$ as well as $\hR_x=\hR_x^{\dagger}$. 

Starting with the first term in \eqref{derivativeLS_wrt_theta}, substituting $\R_x$ with \eqref{blockmatrixRx},
\begin{equation}\label{derivativeRxFirstterm}
\begin{aligned}
&\frac{\partial}{\partial\theta_d}\Tr\left(\R_x\R_x\right)=\\
&\frac{\partial}{\partial\theta_d}\sum_{d_1,d_2=1}^{D}\Tr\big(\left(\F(\theta_{d_1})\otimes\A_{d_1}\right)\left(\F(\theta_{d_2})\otimes\A_{d_2}\right)\big)=\\
&\frac{\partial}{\partial\theta_d}\sum_{d_1,d_2=1}^{D}\Tr\big(\left(\F(\theta_{d_1})\F(\theta_{d_2})\right)\otimes\left(\A_{d_1}\A_{d_2}\right)\big)=\\
&\frac{\partial}{\partial\theta_d}\sum_{d_1,d_2=1}^{D}\Tr\big(\F(\theta_{d_1})\F(\theta_{d_2})\big)\Tr\big(\ua_{d_1}\ua^{\dagger}_{d_1}\ua_{d_2}\ua^{\dagger}_{d_2}\big)=\\
&\frac{\partial}{\partial\theta_d}\Bigg(2\hspace{-0.04cm}\sum_{\substack{k=1 \\ k\neq d}}^{D}\hspace{-0.04cm}\left|\ua^{\dagger}_{d}\ua_{k}\right|^2\hspace{-0.04cm}\big(\uc(\theta_{d})^{\tps}\uc(\theta_{k})\big)^2\hspace{-0.04cm}+\hspace{-0.04cm}\underbrace{\overbrace{\norm{\uc(\theta_d)}^4_{2}}^{=2}\norm{\ua_d}^4_{2}}_{\Rightarrow \text{independent of $\theta_d$}}\Bigg)\hspace{-0.04cm}=\\
&4\sum_{\substack{k=1 \\ k\neq d}}^{D}\left|\ua^{\dagger}_{d}\ua_{k}\right|^2\Big(1+\cos(\theta_{k})\cos(\theta_{d})+\sin(\theta_{k})\sin(\theta_{d})\Big)\\
&\quad\quad\quad\quad\;\;\;\;\cdot\Big(\cos(\theta_{d})\sin(\theta_{k})-\sin(\theta_{d})\cos(\theta_{k})\Big)=\\
&\alpha_1\cos(\theta_d)-\beta_1\sin(\theta_d)+\gamma_1\cos(2\theta_d)-\delta_1\sin(2\theta_d),
\end{aligned}
\end{equation}
where we have defined
\begin{equation}\label{constantsfirstterm}
\begin{gathered}
\alpha_1\triangleq4\sum_{\substack{k=1 \\ k\neq d}}^{D}\left|\ua^{\dagger}_{d}\ua_{k}\right|^2\sin(\theta_k),\; \beta_1\triangleq4\sum_{\substack{k=1 \\ k\neq d}}^{D}\left|\ua^{\dagger}_{d}\ua_{k}\right|^2\cos(\theta_k),\\
\gamma_1\triangleq2\sum_{\substack{k=1 \\ k\neq d}}^{D}\left|\ua^{\dagger}_{d}\ua_{k}\right|^2\sin(2\theta_k),\; \delta_1\triangleq2\sum_{\substack{k=1 \\ k\neq d}}^{D}\left|\ua^{\dagger}_{d}\ua_{k}\right|^2\cos(2\theta_k).
\end{gathered}
\end{equation}
Moving to the second term in \eqref{derivativeLS_wrt_theta}, observe first that
\begin{equation}\label{derivativeRx}
\nabla_{\theta_d}\R_x\hspace{-0.03cm}=\hspace{-0.03cm}\frac{\partial \R_x}{\partial\theta_d}\hspace{-0.03cm}=\hspace{-0.03cm}\frac{\partial}{\partial\theta_d}\sum_{d'=1}^{D}{\F(\theta_{d'})\otimes\A_{d'}}\hspace{-0.03cm}=\hspace{-0.03cm}\frac{\partial\F(\theta_d)}{\partial\theta_d}\otimes\A_d,
\end{equation}
where
\begin{equation}\label{gradFtheta}
\frac{\partial\F(\theta_d)}{\partial\theta_d}\hspace{-0.05cm}=\hspace{-0.05cm}\begin{bmatrix}
0\hspace{-0.025cm}&\hspace{-0.025cm}-\sin(\theta_d)\hspace{-0.025cm}&\hspace{-0.025cm}\cos(\theta_d)\\
-\sin(\theta_d)\hspace{-0.025cm}&\hspace{-0.025cm}-\sin(2\theta_d)\hspace{-0.025cm}&\hspace{-0.025cm}\cos(2\theta_d)\\
\cos(\theta_d)\hspace{-0.025cm}&\hspace{-0.025cm}\cos(2\theta_d)\hspace{-0.025cm}&\hspace{-0.025cm}\sin(2\theta_d)\end{bmatrix}\hspace{-0.05cm}\triangleq\hspace{-0.05cm}
\nabla_{\theta_d}\F(\utheta).
\end{equation}
Substituting \eqref{gradFtheta} into \eqref{derivativeRx}, and using \eqref{covmatofxest}, we now have
\begin{equation}\label{derivativeRxSecondterm}
\begin{aligned}
&\frac{\partial}{\partial\theta_d}\Tr\left(\R_x\hR_x\right)=\Tr\left(\nabla_{\theta_d}\R_x\hR_x\right)=\\
&\Tr\left({\begin{bmatrix}
0\hspace{-0.025cm}&\hspace{-0.025cm}-\sin(\theta_d)\A_d\hspace{-0.025cm}&\hspace{-0.025cm}\cos(\theta_d)\A_d\\
-\sin(\theta_d)\A_d\hspace{-0.025cm}&\hspace{-0.025cm}-\sin(2\theta_d)\A_d\hspace{-0.025cm}&\hspace{-0.025cm}\cos(2\theta_d)\A_d\\
\cos(\theta_d)\A_d\hspace{-0.025cm}&\hspace{-0.025cm}\cos(2\theta_d)\A_d\hspace{-0.025cm}&\hspace{-0.025cm}\sin(2\theta_d)\A_d\end{bmatrix}}\hR_x\right)=\\
&\alpha_2\cos(\theta_d)-\beta_2\sin(\theta_d)+\gamma_2\cos(2\theta_d)-\delta_2\sin(2\theta_d),
\end{aligned}
\end{equation}
where we have defined
\begin{equation}\label{constantssecondterm}
\begin{gathered}
\alpha_2\triangleq 2\ua^{\dagger}_d\hR_x^{(1,3)}\ua_d,\quad\beta_2\triangleq 2\ua^{\dagger}_d\hR_x^{(1,2)}\ua_d,\\
\gamma_2\triangleq 2\ua^{\dagger}_d\hR_x^{(2,3)}\ua_d,\;\;\delta_2\triangleq 2\ua^{\dagger}_d\left(\hR_x^{(2,2)}-\hR_x^{(3,3)}\right)\ua_d.
\end{gathered}
\end{equation}
Finally, substituting \eqref{derivativeRxFirstterm} and \eqref{derivativeRxSecondterm} into \eqref{derivativeLS_wrt_theta}, we obtain the LS equations for $\theta_d$, 
\begin{equation}\label{LSequationsthetad}
\begin{aligned}
&\frac{\partial \widetilde{C}_{\LS}}{\partial\theta_d}=\alpha\cos(\theta_d)-\beta\sin(\theta_d)+\gamma\cos(2\theta_d)-\delta\sin(2\theta_d)=\\
&\cos(\theta_d)\left(\alpha-\beta\tan(\theta_d)\right)+\cos(2\theta_d)\left(\gamma-\delta\tan(2\theta_d)\right)=0,
\end{aligned}
\end{equation}
for every $d\in\{1,\ldots,D\}$, where we have defined
\begin{equation}\label{constantsLSequation}
\begin{gathered}
\alpha\triangleq\alpha_1-2\alpha_2,\quad\beta\triangleq 2\beta_2-\beta_1,\\
\gamma\triangleq\gamma_1-2\gamma_2,\quad\delta\triangleq 2\delta_2-\delta_1.
\end{gathered}
\end{equation}
For the alternating LS equations, solving for $\theta_{d}$ while fixing all the other parameters (as in \eqref{costAJDwrtsingletheta}), we substitute $\hA$ and $\{\widehat{\theta}_k\}_{k\neq d}$ with $\A$ and $\{\theta_k\}_{k\neq d}$, resp., everywhere in \eqref{constantsLSequation}.

At this point, introducing the (invertible) transformation $\tau\triangleq\tan\left(\frac{\theta_d}{2}\right)$ and using basic trigonometric identities, we have
\begin{align}
&\cos\left(\theta_d\right)=\frac{1-\tau^2}{1+\tau^2},\;\cos\left(2\theta_d\right)=\frac{(1-\tau^2)^2-4\tau^2}{(1-\tau^2)^2+4\tau^2},\label{transformationtantheta1}\\
&\tan\left(\theta_d\right)=\frac{2\tau}{1-\tau^2},\;\tan\left(2\theta_d\right)=\frac{4\tau(1-\tau^2)}{1-6\tau^2+\tau^4}.\label{transformationtantheta2}
\end{align}
Finally, substituting \eqref{transformationtantheta1}--\eqref{transformationtantheta2} into \eqref{LSequationsthetad}, we obtain
\begin{equation}\label{transfomredLSeq}
(3\gamma+\alpha)\tau^4+2\beta\tau^3+2\gamma\tau^2+(4\delta+2\beta)\tau-(\alpha+\gamma)=0.
\end{equation}

\section{Computation of the Score and the CRLB}\label{AppD}
In order to obtain closed-form expressions of the score w.r.t.\ each element of the vector of unknowns $\uvarphi$, using the chain rule, we may alternatively use \eqref{score_wrt_varphi}. Hence, our goal now is to compute the gradient of $\R_y$ w.r.t. each element, and the gradient the log-likelihood $\mathcal{L}(\uvarphi)$ w.r.t.\ $\R_y$.

Starting with the gradient of $\R_y$ w.r.t.\ the real and imaginary parts of $\{A_{md}\}$, using \eqref{covaiancematrixofr} and \eqref{blockmatrixRx} we have
\begin{gather}
\nabla_{\Re\{A_{md}\}}\R_y=\frac{\partial}{\partial\Re\{A_{md}\}}\left(\sum_{d=1}^{D}{\F(\theta_d)\otimes\A_d}+\sigma_v^2\I_{3M}\right)\nonumber\\
=\F(\theta_d)\otimes\frac{\partial\ua_d\ua_d^{\dagger}}{\partial\Re\{A_{md}\}}=\F(\utheta)\otimes\left(\ue_m\ua_d^{\dagger}+\ua_d\ue_m^{\tps}\right),\label{gradrealA_App}
\end{gather}
for all $m\in\{1,\ldots,M\}$, and in the same fashion
\begin{equation}\label{gradimagA_App}
\nabla_{\Im\{A_{\tilde{m}d}\}}\R_y=\jmath\cdot\F(\utheta)\otimes\left(\ue_{\tilde{m}}\ua_d^{\dagger}-\ua_d\ue_{\tilde{m}}^{\tps}\right),
\end{equation}
for all $\tilde{m}\in\{2,\ldots,M\}$ (recall $A_{1d}\in\Rset_{\geq0}$), where \eqref{gradrealA_App}--\eqref{gradimagA_App} are for all $d\in\{1,\ldots,D\}$. Next, using \eqref{covaiancematrixofr}, \eqref{derivativeRx} and \eqref{gradFtheta}, the gradient of $\R_y$ w.r.t.\ $\theta_d$ is given by (for all $d\in\{1,\ldots,D\}$)
\begin{equation}\label{gradthetad_App}
\nabla_{\theta_d}\R_y\hspace{-0.02cm}=\hspace{-0.02cm}\nabla_{\theta_d}\R_x\hspace{-0.02cm}=\hspace{-0.02cm}\nabla_{\theta_d}\F(\utheta)\otimes\A_d, \forall \;d\in\{1,\ldots,D\}.
\end{equation}
Lastly, the gradient of $\R_y$ w.r.t.\ $\sigma_v^2$ reads
\begin{equation}\label{gradsigmav_App}
\nabla_{\sigma_v^2}\R_y=\frac{\partial}{\partial\sigma_v^2}\left(\sum_{d=1}^{D}{\F(\theta_d)\otimes\A_d}+\sigma_v^2\I_{3M}\right)=\I_{3M}.
\end{equation}

Having obtained $\nabla_{\varphi_i}\R_y$ for all $i\in\{1,\ldots,K_{\varphi}\}$, namely the gradient w.r.t.\ all the unknowns, we move proceed to the gradient of log-likelihood w.r.t.\ $\R_y$. Recall from \eqref{loglikelihood} that
\begin{equation*}
\mathcal{L}(\uvarphi)\triangleq -T\cdot\left(\log\det\R_y + \Tr\left(\hR_y\R_y^{-1}\right)\right)+c,
\end{equation*}
where $c$ is independent of $\uvarphi$ and $\R_y$. Hence, based on well-known matrix functions derivatives (e.g., \cite{petersen2008matrix}), using 
\begin{align}
&\frac{\partial\log\det\R_y}{\partial\R_y} = \left(\R_y^{-1}\right)^{\tps}\in\Cset^{3M\times3M},\\
&\frac{\partial\Tr\left(\hR_y\R_y^{-1}\right)}{\partial\R_y}=-\left(\R_y^{-1}\hR_y\R_y^{-1}\right)^{\tps}\in\Cset^{3M\times3M},
\end{align}
we obtain
\begin{align}\label{score_wrt_covmatR_App}
\hspace{-0.1cm}\nabla_{\bm{R}_y}\mathcal{L}&=-T\cdot\frac{\partial}{\partial\R_y}\left(\log\det\R_y + \Tr\left(\hR_y\R_y^{-1}\right)\right)\\
&=-T\cdot\left[\R_y^{-1}\left(\I_{3M}-\hR_y\R_y^{-1}\right)\right]^{\tps}\in\Cset^{3M\times3M}.
\end{align}
Substituting \eqref{gradrealA_App}--\eqref{gradsigmav_App} and \eqref{score_wrt_covmatR_App} into \eqref{score_wrt_varphi}, we obtain closed-form expressions for the score. Further, substituting \eqref{gradrealA_App}--\eqref{gradsigmav_App} into \eqref{FIM_CN_element}, we obtain closed-form expressions for the FIM elements, which, upon inversion, gives the CRLB.

\section{Asymptotic ML Estimation based on $\hR_y$}\label{AppE}
As mentioned in Subsection \ref{subsec:suboptimality}, by virtue of the central limit theorem,
\begin{equation}\label{CLT_cov_mat_App}
\text{vec}^*\left(\hR_y\right)\triangleq\hur_y\xrightarrow[\quad\;]{d}\mathcal{CN}\left(\ur_y,\mGamma_{\varepsilon},\C_{\varepsilon}\right).
\end{equation}
Our goal here is to show that ML estimation of $\uvarphi$ based only on $\hur_y$, asymptotically amounts to OWNLLS \eqref{AsymptoticMLEviaOWLS}. To this end, we first derive closed-form expressions of the elements of the covariance and pseudo-covariance matrices $\mGamma_{\varepsilon}$ and $\C_{\varepsilon}$, resp.

By definition, $\mGamma_{\varepsilon}, \C_{\varepsilon}$ are also the covariance and pseudo-covariance matrices of the vector of errors $\uvarep=\text{vec}\left(\Ep\right)$ in estimation of $\R_y$. Therefore, we may equivalently compute the covariance and pseudo-covariance of $\Ep$, which explicitly determine $\mGamma_{\varepsilon}$ and $\C_{\varepsilon}$ by the injective $\text{vec}^*(\cdot)$ mapping. Thus,
\begin{equation}\label{firststep}
\begin{gathered}
\Eset\left[\ep_{ij}\ep_{k\ell}^*\right]=\Eset\left[\widehat{R}_{y_{ij}}\widehat{R}^*_{y_{k\ell}}\right]-R_{y_{ij}}R^*_{y_{k\ell}}=\\
\frac{1}{T^2}\sum_{t_1,t_2=1}^{T}{\Eset\left[y_i[t_1]y_j^*[t_1]y_k^*[t_2]y_\ell[t_2]\right]}-R_{y_{ij}}R^*_{y_{k\ell}}.
\end{gathered}
\end{equation}
Using the fact that $\uy[t]$ is proper, we have
\begin{equation}\label{y_is_proper}
\begin{gathered}
\Eset\left[y_i[t_1]y_j[t_2]\right]=\Eset\left[y_i^*[t_1]y^*_j[t_2]\right]=0,\\
\forall i,j\in\{1,\ldots,M\},\; \forall t_1,t_2\in\{1,\ldots,T\},
\end{gathered}
\end{equation}
hence we may write the summand in \eqref{firststep} as
\begin{multline}\label{fourth_moment_formula}
\Eset\left[y_i[t_1]y_j^*[t_1]y_k^*[t_2]y_\ell[t_2]\right]=\\
\begin{cases}
\kappa_y[i,j,\ell,k]+R_{y_{ij}}R_{y_{k\ell}}^*+R_{y_{ik}}R_{y_{j\ell}}^*, & t_1=t_2\\
R_{y_{ij}}R_{y_{k\ell}}^*, & t_1\ne t_2
\end{cases}.
\end{multline}
Substituting \eqref{fourth_moment_formula} into \eqref{firststep}, and repeating for $\Eset\left[\ep_{ij}\ep_{k\ell}\right]$ with exactly the same technique, we obtain after simplification
\begin{align}
\Eset\left[\ep_{ij}\ep_{k\ell}^*\right]&=\frac{1}{T}\left(\kappa_y[i,j,\ell,k]+R_{y_{ik}}R^*_{y_{j\ell}}\right),\label{errorcovariance1App}\\
\Eset\left[\ep_{ij}\ep_{k\ell}\right]&=\frac{1}{T}\left(\kappa_y[i,j,k,\ell]+R_{y_{i\ell}}R^*_{y_{jk}}\right),\label{errorcovariance2App}
\end{align}
for all $i,j,k,\ell\in\{1,\ldots,M\}$. Note that in the particular case of CN sources, it follows that $\uy[t]$ is CN, and therefore $\kappa_y[i,j,\ell,k]$ vanishes for all $i,j,k,\ell\in\{1,\ldots,M\}$, which gives \eqref{covandpseudocovGaussian}. Thus, \eqref{errorcovariance1App} and \eqref{errorcovariance2App} are the closed-form expressions for all the associated elements (in compliance with the $\text{vec}^*(\cdot)$ mapping) of $\mGamma_{\varepsilon}$ and $\C_{\varepsilon}$, resp., as required.

Next, we turn to the asymptotic ML estimation of $\uvarphi$ based on $\hR_y$, whose asymptotic distribution is prescribed in \eqref{CLT_cov_mat_App}. First, notice that, according to \eqref{errorcovariance1App}--\eqref{errorcovariance2App}, all the covariances and pseudo-covariances are a multiplication of $\tfrac{1}{T}$ by a factor independent of $T$. Accordingly, we define
\begin{equation}\label{covariancesdecomposition}
\tmGamma_{\varepsilon}\triangleq T\cdot\mGamma_{\varepsilon},\quad \tC_{\varepsilon}\triangleq T\cdot\C_{\varepsilon},
\end{equation}
such that $\tmGamma_{\varepsilon}$ and $\tC_{\varepsilon}$ depend only on the elements of $\R_y$, and more specifically, are {\myfontb\emph{independent}} of $T$. With these notations, we further define
\begin{equation}\label{errorCovtilde}
\tR_{\varepsilon}\triangleq\begin{bmatrix}
\tmGamma_{\varepsilon} & \tC_{\varepsilon} \\
\tC_{\varepsilon}^* & \tmGamma_{\varepsilon}^*\end{bmatrix}\;\Longrightarrow\;\R_{\varepsilon}=\frac{1}{T}\cdot\tR_{\varepsilon},
\end{equation}
and the auxiliary matrices (to be used shortly)
\begin{equation}\label{auxiliarymatrixP}
\P_{\varepsilon}\triangleq\mGamma^*_{\varepsilon}-\C^*_{\varepsilon}\mGamma^{-1}_{\varepsilon}\C_{\varepsilon}\;\Longrightarrow\;\tP_{\varepsilon}\triangleq T\cdot\P_{\varepsilon},
\end{equation}
which implies that $\tR_{\varepsilon}$ and $\tP_{\varepsilon}$ are also independent of $T$.

Now, using the asymptotic distribution \eqref{CLT_cov_mat_App}, the probability density function of $\hur_y$ reads
\begin{equation}\label{asymptoticCNpdfApp}
p_{\hat{\text{\boldmath $r$}}_y}\left(\hur_y;\uvarphi\right)\triangleq{\frac{e^{-\frac{1}{2}\left[\uvarep^{\dagger}\;\;\uvarep^{\tps}\right]\R_{\varepsilon}^{-1}\left[\uvarep^{\tps}\;\;\uvarep^{\dagger}\right]^{\tps}}}{\pi^{K_r}\sqrt{\det\left(\mGamma_{\varepsilon}\right)\det\left(\P_{\varepsilon}\right)}}}.
\end{equation}
Thus, by definition, the MLE $\huvarphi_{\ML}$ of $\uvarphi$ based on $\hur_y$, is asymptotically given by
\begin{align}
&\underset{\text{{\boldmath $\varphi$}$\in\Rset^{K_{\varphi}\times1}$}}{\argmax}p_{\hat{\text{\boldmath $r$}}_y}\left(\hur_y;\uvarphi\right)\underset{(i)}{=}\underset{\text{{\boldmath $\varphi$}$\in\Rset^{K_{\varphi}\times1}$}}{\argmax}\log p_{\hat{\text{\boldmath $r$}}_y}\left(\hur_y;\uvarphi\right)\underset{(ii)}{=}\nonumber\\
&\underset{\text{{\boldmath $\varphi$}$\in\Rset^{K_{\varphi}\times1}$}}{\argmin}\log\det\mGamma_{\varepsilon}\P_{\varepsilon}+\left[\uvarep^{\dagger}\;\;\uvarep^{\tps}\right]\R_{\varepsilon}^{-1}\left[\uvarep^{\tps}\;\;\uvarep^{\dagger}\right]^{\tps}\underset{(iii)}{=}\nonumber\\
&\underset{\text{{\boldmath $\varphi$}$\in\Rset^{K_{\varphi}\times1}$}}{\argmin}\underbrace{\log\det\tmGamma_{\varepsilon}\tP_{\varepsilon}}_{\text{independent of $T$}}+\underbrace{T\cdot\left[\uvarep^{\dagger}\;\;\uvarep^{\tps}\right]\tR_{\varepsilon}^{-1}\left[\uvarep^{\tps}\;\;\uvarep^{\dagger}\right]^{\tps}}_{\text{linear in $T$}}\underset{(iv)}{\approx}\nonumber\\
&\underset{\text{{\boldmath $\varphi$}$\in\Rset^{K_{\varphi}\times1}$}}{\argmin}\left[\uvarep^{\dagger}\;\;\uvarep^{\tps}\right]\R_{\varepsilon}^{-1}\left[\uvarep^{\tps}\;\;\uvarep^{\dagger}\right]^{\tps}\triangleq\huvarphi_{\OWNLLS},
\end{align}
which is the OWNLLS estimate, and we have used the fact that $\log$ is an increasing monotonic function in $(i)$, omitted irrelevant constants w.r.t.\ $\uvarphi$ in $(ii)$, as well as in $(iii)$ due to
\begin{align*}\label{logpropertyApp1}
&\log\det\mGamma_{\varepsilon}\P_{\varepsilon}=\log\det\left(\frac{1}{T^2}\cdot\tmGamma_{\varepsilon}\tP_{\varepsilon}\right)=\\
&\log\left(T^{-2K_r}\det\tmGamma_{\varepsilon}\tP_{\varepsilon}\right)=\underbrace{-2K_r\cdot\log T}_{\text{independent of \boldmath$\varphi$}}+\log\det\tmGamma_{\varepsilon}\tP_{\varepsilon},
\end{align*}
and the approximation $(iv)$ holds for a sufficiently large sample size $T\gg\left(\log\det\tmGamma_{\varepsilon}\tP_{\varepsilon}\right)/\left(\left[\uvarep^{\dagger}\;\;\uvarep^{\tps}\right]\tR_{\varepsilon}^{-1}\left[\uvarep^{\tps}\;\;\uvarep^{\dagger}\right]^{\tps}\right)$.

\bibliography{Bibfile}
\bibliographystyle{unsrt}

\end{document}